\documentclass[letterpaper, DIV=13]{scrartcl}   
\usepackage{authblk}

\usepackage{amssymb, amsthm, mathtools,bbm,color}
\usepackage[square,sort,comma,numbers]{natbib}
\usepackage{dsfont} 
\usepackage{amsmath}
\usepackage{cancel}   

\def\1{\mathbbm{1}}
\def\N{\mathbbm{N}}
\def\Z{\mathbbm{Z}}
\def\R{\mathbbm{R}}

\newcommand{\sgn}{\operatorname{sign}}
\newcommand{\dist}{\operatorname{dist}}
\def\tr{\mathop{\textrm{tr}}\nolimits}

\newtheorem{theorem}{Theorem}
\newtheorem{lem}[theorem]{Lemma}
\newtheorem{cor}[theorem]{Corollary}

\newtheorem{proposition}[theorem]{Proposition}

\numberwithin{equation}{section}
\numberwithin{theorem}{section}
\usepackage{graphicx}	
\usepackage{amssymb}

\def\S {\textbf{S}}
\def\be{\begin{equation}}
\def\ee{\end{equation}} 
\def\ket#1{\vert#1\rangle}
\def\bra#1{\langle #1\vert}

\title{Dimerization  and N\'eel order in different quantum spin chains through a shared loop representation  }

\author[1]{Michael Aizenman}

\author[2,3]{Hugo Duminil-Copin}

\affil[1]{{\small Departments of Physics and Mathematics, 
Princeton University, 
Princeton, NJ 08544, USA}}

\affil[2]{{\small IHES, 
91440 Bures-sur-Yvette, France}}
\affil[3]{{\small Dpt.\ de Math\'ematiques, 
 Universit\'e de Geneve, 
 1211 Geneve 4, Switzerland}}
 
\author[4,5]{Simone Warzel}
\affil[4]{{\small
Zentrum Mathematik, TU M\"unchen, 
85747 Garching, Germany}}
\affil[5]{{\small Munich Center for Quantum Science and Technology}}
\date{\small June 11, 2020}

\begin{document}
\maketitle

\begin{abstract}
 The ground-states of the spin-$ S $ antiferromagnetic chain~$H_\textrm{AF}$ with a projection-based interaction and the spin-$ 1/2$ XXZ-chain~$ H_\textrm{XXZ} $ at anisotropy parameter $\Delta=\cosh(\lambda) $ share a common loop representation in terms of a two-dimensional functional integral which is similar to the classical planar $Q$-state Potts model at 
 $ \sqrt Q= 2S+1 =2\cosh(\lambda)$.  
The multifaceted relation  
is used here to  directly relate the distinct forms of  translation symmetry breaking which are manifested in the ground-states of these two models:    dimerization for $H_\textrm{AF}$ at all $S>  1/2$,  and N\'eel order for $ H_\textrm{XXZ} $ at $\lambda >0$.   The results presented include: i)~a translation to the above quantum spin systems of the results which were recently proven by Duminil-Copin-Li-Manolescu for a broad class of two-dimensional random-cluster models,  and ii)~a short proof of the symmetry breaking in a manner similar to the recent structural proof by Ray-Spinka of the discontinuity of the phase transition for $Q>4$.   
Altogether, the quantum manifestation of the change between  $Q=4$ and  $Q>4$ is a transition from a  gapless ground-state to  a pair of gapped  and extensively distinct  ground-states.
\end{abstract}

\section{Introduction}

The focus of this work is the structure of the ground-states in two families of antiferromagnetic quantum spin chains, each of which includes the spin-$1/2$ Heisenberg anti-ferromagnet as a special case.    In the infinite volume limit, with the exception of their common root, in both cases the systems exhibit symmetry breaking at the level of ground-states.  The physics underlying  the phenomenon is different. In one case it is extensive quantum frustration which causes dimerization with is expressed in spatial energy oscillations.  In the other case, the Hamiltonian is frustration free and the symmetry breaking is expressed in long-range N\'eel order.  Yet, in mathematical terms both phenomena are analyzable through a common random loop representation.  Curiously, a similar loop system appears also  as the auxiliary scaffolding of a classical planar $Q$-state Potts models for which the symmetry breaking relates to a discontinuity in the order parameter.

The models under consideration have been studied extensively, and hence the specific results we discuss  may be regarded as known, at one level or another.   The techniques which have been applied for the purpose  include numerical works, Bethe ansatz calculations~\cite{BBa,BBb,Aff,baxter:1973,Bax16,klu}, and contour expansions~\cite{NU}.   The validity of Bethe ansatz  calculations for similar systems has recently received support through a careful mathematical analysis~\cite{DGHMT}.  The results presented here are based on non-perturbative structural arguments.   
They may be worth presenting since  in the models considered such arguments allow full characterization of the conditions under which the symmetry breaking occurs, as well as other qualitative features of the model's ground-states.    
The relation between the models may be of intrinsic interest.  At the mathematical level it  plays an essential role in the non-perturbative proof of symmetry breaking which is the main result presented here.

\subsection{Antiferromagnetic $SU(2S+1)$ invariant spin chains  with projection based   interaction} 

The most basic quantum object has a two-dimensional complex state space, spanned by the two orthogonal vectors $  \ket{+} \equiv \left( \begin{matrix}  1 \\   0  \end{matrix}\right) $ and  $ \ket{-} \equiv \left( \begin{matrix} 0   \\ 1   \end{matrix}\right)  $.   The self-adjoint operators on this space (which has the structure of $\mathbb{C}^2$) are linear combinations of the  three Pauli-spin matrices 
$ \pmb{\tau} = (\tau^x,\tau^y , \tau^z) $, 
\be\label{eq:Pauli} 
\tau^x := \left( \begin{matrix} 0 & 1 \\ 1 & 0  \end{matrix}\right) \, , \quad \tau^y := \left( \begin{matrix} 0 & -i \\ i & 0  \end{matrix} \right) \, , \quad \tau^z := \left( \begin{matrix} 1 & 0 \\ 0 & -1  \end{matrix} \right)\, . 
\ee
Of particular interest is the triplet of spin operators 
$ \S = (S^x,S^y , S^z) $ with $S^\alpha = \frac 12 \tau^\alpha$, $\alpha = x,y,z$.  These span the Lie algebra of the group $SU(2)$ and satisfy the commutation relation  
 \be \label{comm_rel}  [ S^x, S^y ] = i  S^z \,. 
 \ee

For higher spin systems the Hilbert spaces  of states are given by $\mathbb{C}^{2S+1} $  in which one finds the $2S+1$  dimensional representations of the Lie algebra commutation relations \eqref{comm_rel}, with $S\in \N \cup (\N -\frac 12) $.  A convenient basis is provided by the eigenvectors of $S^z$, satisfying 
\be 
S^z\  \ket{\ m\  }\  =\  m\   \ket{\  m\ } \, , \quad  m \in \{-S, -S+1,\dots, S\} .
 \ee 
In this terminology, the above binary spin system corresponds to $S=1/2$, and the states $\ket{+}$ and $\ket{-}$  are the eigenstates of $S^z$ at values $m= - \frac 12, + \frac 12 $.

Our spin chains are arrays of $ 2 L $ spins indexed by $ \Lambda_L:= \{-L+1,\dots,L\} $.  The corresponding  state space is the tensor product Hilbert space $\mathcal{H}_L=  \bigotimes_{v\in \Lambda_L} \mathbb{C}^{2S+1} $.  The  single spin operators are lifted to it by setting $ \pmb{\tau}_u := \mathbbm{1} \otimes \cdots \mathbbm{1} \otimes \pmb{\tau}  \otimes  \mathbbm{1}  \cdots \mathbbm{1}  $, which acts non-trivially only in the  tensor product's $u $th component. 

Lifted to the two component product space, the above Dirac notation of states takes the form 
\be 
| m , m' \rangle_{u,v} :=  | m  \rangle_{u} \otimes   | m'  \rangle_{v}  \quad \quad  {}_{u,v}\langle m , m' |  :=   {}_u\langle m | \otimes   {}_v\langle  m' |  \,. 
\ee 
Correspondingly, we shall use the following notation for operators acting in the corresponding two-component factor of $\mathcal{H}_L$
\be \label{ket_bra} 
 \left( | m , m' \rangle\langle n, n' | \right)_{u,u+1} := 
\mathbbm{1}  \otimes \dots  \mathbbm{1} \otimes \left( | m  \rangle\langle n | \right)_u \otimes \left(  | m'  \rangle\langle n' | \right)_{u+1} \otimes\mathbbm{1}   \ldots \otimes \mathbbm{1}  
\ee

Our discussion will focus on different extensions of the quantum Heisenberg 
antiferromagnetic spin model, which is an array of spins with the nearest-neighbor interaction energy proportional to $\S_u \cdot \S_{u+1} $.  For $S=1/2$ this can be alternatively written as
\be \label{eq:AF0}
H_\textrm{AF}^{(L)} :=  \sum_{u=-L+1} ^{L-1}  [ 2 \ \pmb{\tau}_u\cdot \pmb{\tau}_{u+1} -1/2]  \ 
= \    -  2    \sum_{u=-L+1} ^{L-1}  P^{(0)}_{u,u+1}    \, . 
\ee
with $\pmb{\tau}_u\cdot \pmb{\tau}_{u+1} = \sum_{\alpha = x,y,z} \tau_u^\alpha \ \tau_{u+1}^\alpha$ and 
$P^{(0)}_{u,u+1} =\left( \ket{D} \bra{D}\right)_{u,u+1} $ the orthogonal projection onto the state 
\be \label{eq:singlett}
\ket{D} := (\ket{+,-}- \ket{-,+})/\sqrt 2\, ,
\ee 
in the corresponding two-spin space.   This state is of some interest:  it is the only one which is annihilated by each component of the combined spin operator $\S_u +\S_{u+1} $, and it also  maximizes the entanglement between the two components.

The two expressions of the spin 1/2 Hamiltonian which are presented in \eqref{eq:AF0}  suggest slightly different extensions  to higher values of the spin $ S \in \mathbb{N} /2 $.    The one on which we  focus here is
\be  \label{H_AF} 
H_\textrm{AF}^{(L)} :=  \ -\sum_{u=-L+1} ^{L-1}  
(2S+1)   P^{(0)}_{u,u+1}   \, ,
\ee
with  $P^{(0)}_{u,u+1} $   the 
rank-one projection in the two spin space  $ \mathbb{C}^{2S+1} \otimes  \mathbb{C}^{2S+1} $ 
onto on the subspace which is invariant under rotations generated by $\S_u+\S_{u+1}$, i.e.  the joint kernel of $S_u^\alpha +\S_{u+1}^\alpha$  ($\alpha = x,y,z$).    
For any $S\in \N/2$ this  operator is given by\footnote{The projection $P_{u,v}^{(0)}$ can also be expressed as a polynomial of degree $2S$ in $\pmb{S}_u\cdot \pmb{S}_{v}$, for instance $P_{u,v}^{(0)} = ((\pmb{S}_u\cdot \pmb{S}_{v} )^2-1)/3$ for $S=1$.}
\begin{eqnarray}\label{eq:pro}
P_{u,v}^{(0)} &:=& 1\left[\left| \pmb{S}_u + \pmb{S}_v \right| = 0 \right] = \frac{1}{2S+1} \sum_{m,m'=-S}^S (-1)^{m-m'} \left(  \left| m, -m \rangle\langle m', -m' \right| \right)_{u,v}, 
\end{eqnarray}
This model was studied by Affleck~\cite{Aff},  Batchelor and Barber~\cite{BBa,BBb},  Kl\"umper \cite{klu}, Aizenman and Nachtergaele~\cite{AN}, and more recently Nachtergaele and Ueltschi~\cite{NU}.   

The classical analog of a quantum spinor with the state space $\mathbbm{C}^{2S+1}$ is a system whose states are described by a three component vector of length $S$.   Under this correspondence, the classical analog of the projection to the ground-state(s) of $H_\textrm{AF}^{(L)} $ is the restriction to configurations in which each pair of  neighboring spins point in exactly opposite directions, adding  to $\underline 0$.   However, unlike its classical analog, the quantum system exhibits frustration, and that leads to the dimerization phenomenon discussed next.

Each of the two-spin interaction terms in \eqref{H_AF}    is  minimized in the state in which the two spins are coherently intertwined into the unique state in which $|\S_u+\S_v|=0$. 
Yet, a quantum spin cannot be locked into such a state with both its neighbors simultaneously. 
 This effect, which  results in the spin-Peierls instability, is purely quantum  as there is no such restriction for classical  spins.   
(Classical spin models may be driven to frustration by other means, e.g. when placed on a non-bipartite graph with antiferromagnetic interactions, and also on arbitrary graphs under suitably mixed interactions.  Such geometric frustration is shared by their quantum counterparts.)

 \begin{figure}[h]
\begin{center}
\includegraphics [width=0.4 \textwidth]{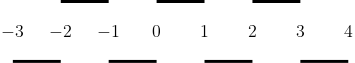}
\caption{The natural  pairing in $\Lambda_L = \{ -L+1, \dots , L-1, L \}$  for $L=3$ and $L= 4$.   Notice the difference at $u=0$.}\label{fig:pairing}
\end{center}
\end{figure}
The naive pairing depicted in Fig.~\ref{fig:pairing} suggests that in finite volume the ground-states' local energy density may not be homogeneous and have a bias triggered by the boundary conditions, i.e.\ the parity of~$L$.  Indeed, 
through approximations, numerical simulations, or the probabilistic representation of \cite{AN} (our preferred method), one may see that the local energy density of the corresponding finite-volume ground-states $ \langle\cdot  \rangle_{L}^{(\textrm{gs})} $  is not homogeneous and satisfies
\begin{equation} 
(-1)^L \big[ \langle P_{2n-1,2n} \rangle_{L}^{(\textrm{gs})} - \langle P_{2n,2n+1} \rangle_{L}^{(\textrm{gs})}  \big] \, >\, 0\,.
\end{equation}
An interesting question is whether this bias persists in the limit $L\to \infty$, in which case in the infinite-volume limit the system has (at least) two distinct ground-states, for which the expectation values of local observables~$F$ are given by   
\begin{equation} \label{ev_odd}
\langle F \rangle_\textrm{even} \  := \ \lim_{\substack{L\to \infty \\ L \textrm{even}}}    \langle F \rangle_L^{(\textrm{gs})}  \qquad 
\text{and} 
\qquad
\langle F \rangle_\textrm{odd} \  := \ \lim_{\substack{L\to \infty \\ L \textrm{odd}}}  
    \langle F \rangle_L^{(\textrm{gs})}  \, , 
\end{equation}
 where the limit is interpreted in the weak sense, i.e.\ with $F$ being any 
(fixed) local bounded operator.  These are generated by products of spin operators
\begin{equation}
F_U := \prod_{j=1}^k S_{u_j}^{\alpha_j} \, , \quad u_j \in U \, , \quad \alpha_j \in \{ x,y,z \} \, 
\end{equation}
which are supported in some bounded set $ U \subset \mathbb{Z} $.
In finite-volume, their (imaginary) time-evolved counterparts are given by 
$$ F_U^{(L)}(t) := \ e^{ - t H_\textrm{AF}^{(L)} } \ F_U \ e^{t H_\textrm{AF}^{(L)} } \,  .
$$ 
The corresponding truncated correlations also converge, e.g.~for any fixed $ t \in \mathbb{R} $, 
\begin{equation}\label{eq:trunc}
	\langle F_U(t)  ; F_V\rangle_\textrm{even}  :=  \lim_{\substack{L\to \infty \\ L \textrm{even}}}    \langle F_U^{(L)}(t) F_V \rangle_L^{(\textrm{gs})} -  \langle F_U^{(L)}(t) \rangle_L^{(\textrm{gs})}  \langle  F_V \rangle_L^{(\textrm{gs})} \, ,
\end{equation}
and similarly for $ \langle F_U(t)  ; F_V\rangle_\textrm{odd} $.

The separate convergence of the limits~\eqref{ev_odd} or \eqref{eq:trunc} was established in 
\cite{AN} 
through  probabilistic techniques which are enabled by the loop representation presented below.   This representation also led to the following dichotomy.\footnote{This version of the AN dichotomy is a bit more carefully crafted than in the original work, as the two options stated there need not be mutually exclusive.  However, as \eqref{exp} shows, ipso-facto  they are.}

\begin{proposition}[cf.\ Thm.~6.1 in \cite{AN}]  \label{thm:dichotomy}
For each value of $S\in \mathbb{N}/2$ one of the following holds true:
\begin{enumerate}
\item  The two ground-states $\langle\cdot\rangle_{\textrm even}$ and $\langle\cdot\rangle_{\textrm odd}$ are distinct, each invariant under  the $2$-step shift, each being the $1$-step shift of the other.  Furthermore, their translation symmetry breaking is manifested in energy oscillations, namely, for every $n\in\mathbb N$
\begin{equation} \label{osc}
  \langle P^{(0)}_{2n-1,2n} \rangle_\textrm{even} -  \langle P^{(0)}_{2n,2n+1} \rangle_\textrm{even} > 0\,  .
\end{equation}
\item The even and odd ground-states coincide, and form a translation invariant  ground-state $\langle\cdot\rangle$ with slowly decaying correlations, satisfying  
 \begin{equation} \label{slow}
 \sum_{v\in \mathbb{Z}} |v|\, |\langle \pmb{S}_0 \cdot  \pmb{S}_v \rangle |   =   \infty  \, .
 \end{equation}
  \end{enumerate} 
\end{proposition}

For $ S = 1/2 $ the second  alternative is known to hold (cf.~\cite{Aff_1/2,DST} and references therein).  In this case the model reduces to the quantum Heisenberg antiferromagnet.\footnote{Various features of the model are calculable through the  Bethe ansatz, which was actually developed in that context~\cite{Bethe}.
However, even aside from the extra care which is required for rigorous results, the exact determination of the the long distance asymptotic seems to require other means (cf.~\cite{Aff_1/2, LT03, DST} and references therein).}
In the converse direction, dimerization in this model was established for $S\ge 8$~\cite{NU} through a contour expansion.
The gap between these results is closed  here through a structural proof  that for all $S> 1/2$  the first option holds (regardless of the parity of $2S$). 
\begin{theorem} \label{thm:main1}
For all $S>1/2$:
\begin{enumerate}
\item  the even and odd ground-states,  defined by~\eqref{ev_odd}, differ.  
They are translates of each other, and exhibit the energy oscillation~\eqref{osc}. 
\item there exist $\xi=\xi(S) <\infty $ such that for all $U,V\subset\mathbb Z$ with distance $ \dist(U,V) $ and any~$ t \in \mathbb{R} $: 
\begin{equation} \label{exp}
\vert \langle F_U(t)  ; F_V\rangle_\textrm{even}  \vert \ \leq \  C_{F_U} C_{F_V} \, e^{- (\dist(U,V) + | t| ) /\xi}\, , 
\end{equation}
where $ C_{F_U} $ and $ C_{F_V} $ are invariant under space-time translations of the observables $ F_U , F_V $. 
  \end{enumerate}
\end{theorem}

The proof draws on the progress which was recently made in the study of the related loop models.  In \cite{DGHMT}, the loop representation of the critical $Q$-state Potts model on the square lattice with $Q>4$ was proved to have two distinct infinite-volume measures under which the probability of having large loops is decaying exponentially fast (see \cite{kotecky:1982} for the case of large $Q$). The result was extended in \cite[Theorem 1.4]{DLM18} to a slightly modified version of the loop model that will be redefined in this paper and connected to the spin chains (there, the model is not defined in terms of loops but in terms of percolation, as in Section~\ref{sec:percolation}). More recently, Ray and Spinka \cite{RS} provided an alternative proof of the non-uniqueness of the infinite-volume measures on the square lattice. 

In this article, the inspiring proof of Ray-Spinka is extended to our context to provide a new proof of~1. We believe that this proof is more transparent and conceptual than the one in \cite{DLM18}, and that even though the technique does not directly lead to~2., it illustrates perfectly the interplay between the quantum and classical realms. In fact, a careful analysis of the proofs in the paper of \cite{DLM18}  shows that the argument there relies on two pillars: a  theorem proving a stronger form of Proposition~\ref{thm:dichotomy} (see also \cite{DST,DT19} for versions on the square lattice), in which~1.~of Theorem~\ref{thm:main1} is proved to imply~2., and an argument relying on the Bethe Ansatz showing that~1.~indeed occurs. The adaptation of the Ray-Spinka argument enables us to prove~1.~directly without using the Bethe Ansatz, so that the argument in this paper replace half of the argument in \cite{DLM18}, and that combined with the other half it also implies~2.\\

Let us finally note that under the dimerization scenario, which is now established for its full range ($S > 1/2 $), other physically interesting features follow:
\begin{enumerate}
\itemsep0pt
\item \emph{Spectral gap:} As was argued already in~\cite[Theorem~7.1]{AN}, the exponential decay of truncated correlations~\eqref{exp} in the $ t $-direction implies a non-vanishing spectral gap in the excitation spectrum above the even and odd ground-states. 
\item \emph{Excess spin operators:} 
When the decay of correlations is fast enough so that \eqref{slow} does not hold, in particular under \eqref{exp},  in the even/odd states the spins are organized into tight neutral clusters.  That is manifested in the tightness of the distribution of the block spins 
$S^z_{[a,b]} = \sum_{u\in [a,b]} S^z_u$ (in a sense elaborated in \cite{AGL}).   That is equivalent to the existence of the excess spin operators $ \widehat S_u^z$  with which
\be
\sum_{v=1}^u  S^z_v  =  \widehat S_0^z-\widehat S_u^z\,
\ee 
and such that $\widehat S_u^z$ commutes with the spins in $(-\infty, u]$.  
The quantity $\widehat S_u^z$ can be interpreted as the total spin in $(u,\infty)$, and constructed as  
$\lim_{\varepsilon\downarrow 0} \sum_{v > u } e^{-\varepsilon | u-v| }  S_v^z $ (in the strong-resolvent sense), cf.~\cite[Sec.~6]{AN}.
As was further discussed in~\cite{BachN}, the excess spins play a role  in the classification of  the topological properties of the gapped ground-state phases.
\item \emph{Entanglement entropy:} 
Another general implication of the exponential decay of correlations is a so-called area law (which for chains equates to the boundedness) of the entanglement entropy of the ground-states, see \cite{BH12} for details.
\end{enumerate}

\subsection{The $S=1/2$ antiferromagnetic   $XXZ$ spin chain} 

The second model discussed in this paper is the 
anisotropic XXZ spin-$1/2 $ chain with the Hamiltonian 
\be  \label{H_XXZ}
H^{(L)}_\textrm{\textrm{XXZ}} :=  - \frac{1}{2} \sum_{v=-L+1}^{L-1}\,  
[\tau^x_{v}  \tau^x_{v+1}+  \tau^y_{v}  \tau^y_{v+1}  -  
\Delta \,\, (\tau^z_{v}  \tau^z_{v+1}-1)]
\ee
acting on the Hilbert space $ \mathcal{H}_L  = \bigotimes_{v=-L+1}^L \mathbb{C}^{2} $.   
It consists of Pauli spin matrices~\eqref{eq:Pauli} on $ \mathbb{C}^2 $.
It is convenient to present the anisotropy parameter as
\be  \label{Delta_lambda}
\Delta := \cosh(\lambda)   > 1 \, . 
\ee 
Throughout the paper and unless stated otherwise explicitly, we will take $ \lambda \geq 0 $ the non-negative solution of~\eqref{Delta_lambda}.

The sign and the magnitude of  $\Delta >1$ favor antiferromagnetic order in the ground-state. 
The negative sign in front of the terms involving the $ x $- and $ y $-component of the 
Pauli spin matrices can be flipped through the  unitary transformation
$ U_L  = \exp\big( i \tfrac{\pi}{4} \sum_u (-1)^u \tau_u^z  \big) $. 
It  renders the Hamiltonian  in  the manifestly antiferromagnetic form
\begin{equation} \label{mod_XX}
U_L  H^{(L)}_{\textrm{XXZ}}U_L^*  = 
\frac{1}{2} \sum_{v=-L+1}^{L-1}\,  
\left[    \pmb{\tau}_{v}  \cdot \pmb{\tau}_{v+1}   +  \left( \Delta- 1 \right) \tau^z_{v}  \tau^z_{v+1}  - \Delta\right]  
 \,. 
\end{equation}
The antiferromagnetic   XXZ  chain has been the subject of many works.  Following Lieb's work on interacting Bose gas~\cite{Lieb63}, 
 Yang and Yang gave a justification for the Bethe Ansatz solution of the ground-state in a series of papers~\cite{YY1,YY2} in 1966. The ground-state has long-range order with two period-2 states in the thermodynamic limit, 
each with mean magnetization of alternating direction.  The corresponding 
  N\'eel order parameter ($M_{\text{N\'eel}}$ of ~\eqref{eq:Neel} below)   
  vanishes in the limit $\Delta \downarrow  1$.  
Since the exact solution is not very transparent, there has been interest in obtaining qualitative information by other means, e.g.\ expansions and other rigorous methods.   These typically apply only for large~$ \Delta $.   

Our motivation for returning to the XXZ spin chain is that  
it emerges very naturally in the analysis of the thermal and ground-states of the model $H_\textrm{AF}$. Furthermore, the relation between the two facilitates the proof of the symmetry breaking stated in Theorem~\ref{thm:main1}.   
In the converse relation, this relation is used here to establish symmetry breakdown in the form of N\'eel order of the XXZ ground-state(s) for all $  \Delta > 1 $.

To prove the translation symmetry breaking we consider the pair of finite-volume ground-states for the  Hamiltonian \eqref{H_XXZ} with an added boundary field\footnote{One may expect that in case there is N\'eel order any antisymmetric boundary field would flip the ground-state into one of the extremal states.  However the proof of that is simpler for the case the field's magnitude is at least $|\sinh (\lambda)|$.},  
 i.e.\
\be  \label{XXZ_bc}
H^{(L,\textrm{bc})}_{\textrm{XXZ}} :=  H^{(L)}_{\textrm{XXZ}}
+ \sinh(\lambda) (-1)^L\ \frac{ \tau_{-L+1}^z - \tau_L^z }{2} \times 
\begin{cases} 
+ 1& \mbox{for $\textrm{bc} = +$}\\
-1 & \mbox{for $\textrm{bc} = -$}
\end{cases} .
\ee

As a preparatory statement let us state: 
\begin{proposition} \label{prop:Neel}  
For any $\Delta \geq 1$, in the limit $L\to \infty$ with $ L $ even,  the finite-volume ground-states of the XXZ-spin system with the above boundary terms converge  to  states $\langle \cdot \rangle_+$ and $\langle \cdot \rangle_-$.   Regardless of whether the  two agree,  each is a one-step shift of the other.  
The two states are different if and only if they exhibit N\'eel order, in the sense for all $ n $:
\begin{equation} \label{eq:Neel}
(-1)^{n} \langle \tau_n^z \rangle_+ = -  (-1)^{n} \langle \tau_n^z \rangle_-  = M_{\textrm{N\'eel}}  
\end{equation}
at some  $M_{\text{N\'eel}} \neq 0$.  
\end{proposition} 

Let us emphasize that the system's size is even regardless of the parity of $L$ (the size being equal to $2L$).  The restriction in this theorem to sequences of constant parity is required for the consistency of the effect of the boundary conditions which are specified in~\eqref{XXZ_bc}. \\

Similarly to  Proposition~\ref{thm:dichotomy}, this statement is proven here through  the FKG inequality which is made applicable in a suitable loop representation.   We postpone its proof  to Section~\ref{sec:percolation}, next to the place where it is applied.  
Following is the XXZ-version of the symmetry breaking statement.

\begin{theorem} \label{thm:main2}
For any $ \Delta > 1 $ the construction described in Proposition~\ref{prop:Neel} yields two different ground  states of infinite XXZ-spin chain which differ by a one step shift and satisfy 
\eqref{eq:Neel}. 
\end{theorem}  

Theorem~\ref{thm:main2}  is proven in Section~\ref{sec:percolation} together with Theorem~\ref{thm:main1}.  
In each case the symmetry breaking is initially established through the expectation value of a conveniently defined quasi-local observable.  The conclusion is then boosted  to the more easily recognizable statements presented in the theorems through the preparatory statements of Proposition~\ref{prop:Neel} and respectively Proposition~\ref{thm:dichotomy}. 

\subsection{Seeding the ground-states} 

Infinite-volume ground-states  can be approached through their intrinsic properties (such as the energy-minimizing criterion) or, constructively, as limits of finite-volume ground-state expectation value functionals.  To establish their non-uniqueness,  we shall consider  different sequences of finite volume ground-states, and establish convergence of the expectation value functionals to limits which are  extensively different.    
Equivalently, it suffices to construct a single limiting ground-state which does not have the Hamiltonian's translation symmetry.  A shift (or another symmetry operation)  produces then another ground-state.   We shall take that path  in the discussion of both models. 

The finite-volume ground-states will be constructed through limits of the form 
\be \label{eq:seed}
\langle F \rangle_L^{(\textrm{gs})} =  \lim_{\beta\to \infty} \frac{\bra{\Psi_L}    e^{-\beta H_L/2}  F  e^{-\beta H_L/2}   \ket{\Psi_L}} 
{\bra{\Psi_L}    e^{-\beta H_L}   \ket{\Psi_L}}\,. 
\ee 
with $\ket{\Psi_L}$ a convenient \emph{seeding vector}.   To assure that the limiting functional corresponds to a ground-state (or \emph{the} ground-state if it is unique) one needs to verity that this vector is not annihilated by the  ground-state projection operator $ P_L^{(\textrm{gs})} $. 
That will be established  by verifying that  
\begin{equation}\label{gsproj}
\frac{ \langle  \Psi_L  |  P_L^{(\textrm{gs})}   |   \Psi_L\rangle }{ \dim  P_L^{(\textrm{gs})} } =  \lim_{\beta \to \infty} \frac{\langle \Psi_L | e^{-\beta H_L} |  \Psi_L \rangle }{\tr e^{-\beta H_L} }  > 0 \, . 
\end{equation}
Our choice of the seeding vectors is primarily guided not by the condition \eqref{gsproj}, which is generically satisfied, but rather by the goal of a transparent expression for the expectation value functional.

 In view of the quantum frustration effect, a natural seed vector for the construction of a ground-state for the Hamiltonian $ H_\textrm{AF}^{(L)} $ on an even collection of spins in $\Lambda_L = \{-L+1,\dots,L\}$ is the dimerized state
\begin{eqnarray}\label{eq:dimerstate}
| D_{L} \rangle & := &  \bigotimes_{j=1}^L  \Big(  \sum_{m=-S}^S (-1)^m \,  | m, -m \rangle_{-L+2j-1,-L+2j}\Big)   \notag \\ 
& = & U_{L} \ \bigotimes_{j=1}^L  \Big(  \sum_{m=-S}^S  \,  | m, -m \rangle_{-L+2j-1,-L+2j}\Big)  
\, .
\end{eqnarray}
The subscripts on the vectors indicate on which tensor component of $ \mathcal H_L $ they act. 
The role of the gauge transformation 
\be 
U_L  := \exp\Big( i \tfrac{\pi}{2} \sum_u (-1)^u S_u^z \Big) \, ,
\ee  
expressed in the standard $ z $-basis of the joint eigenstates of~$S^z_u $, $ u \in \Lambda_L $, 
is to ensure non-negativity of the matrix-elements of $ U_L^* e^{-\beta H^{(L)}_\textrm{AF} } U_L $ in the $ z $-basis.
This will enable a probabilistic loop representation of this semigroup presented 
in Section~\ref{sec:FK}. 
From this representation, we will also see that~\eqref{gsproj} is valid for the seed state $ \Psi_L = D_{L} $ at any finite $ L $, cf.~\eqref{eq:conditionAF} below. The standard Perron-Frobenius argument is not applicable in this case.

Applying the semigroup operator $e^{-\beta H_L/2}$ to $\ket{ D_{L} }$,   one gets the expectation-value functional which assigns to each local observable $F$ the value
\be \label{eq:expdimer}
\langle F \rangle_{L,\beta}^{\textrm{(AF)}} \  := \ 
\frac{\bra{ D_{L} } e^{-\beta H^{(L)}_\textrm{AF}/2}  F  e^{-\beta  H^{(L)}_\textrm{AF}/2} \ket{ D_{L} }} 
{\bra{ D_{L} } e^{-\beta  H^{(L)}_\textrm{AF}} \ket{ D_{L} } } \, , 
\ee 
and which converges as $ \beta \to \infty $ to a ground-state expectation $ \langle F \rangle_{L}^{(\textrm{gs})} $.
It is the above expectation-value functional which we study in the proof of Theorem~\ref{thm:main1} by probabilistic means.

To study the N\'eel order of the XXZ-Hamiltonian  we find it convenient to focus on the sequence of constant parity, say even $L$, and use as seed in \eqref{eq:seed}  the  vector 
\begin{equation} \label{N_vector}
| N^{(L)}_\lambda \rangle \ =\   \bigotimes_{j=1}^L 
\Big(e^{-\lambda/2} \ket{ +,- }_{-L+2j -1,-L+2j} + e^{\lambda/2} \ket{ -,+ }_{-L+2j -1,-L+2j} \Big) \, 
\end{equation}
 which is indexed by $ \lambda $. 
Using it, the state $\langle \cdot\rangle_+$ of Proposition~\ref{prop:Neel}  is presentable as the double  limit
\be \label{eq:101}
\langle F \rangle_+^\textrm{(XXZ)}  \ = \ 
 \lim_{\substack{L\to \infty \\ L \textrm{even}}}   \lim_{\beta\to \infty}  
\langle F \rangle_{L,\beta,\lambda}^\textrm{(XXZ,+)} 
\ee
of 
\be 
\langle F \rangle_{L,\beta,\lambda}^{\textrm{(XXZ,+)}} : =
\frac{
\bra{N^{(L)}_\lambda}    e^{-\beta H^{(L,+)}_{\textrm{XXZ}}/2} \,  F\,   e^{-\beta H^{(L,+)}_{\textrm{XXZ}}/2}   \ket{N^{(L)}_\lambda}} 
{\bra{N^{(L)}_\lambda}    e^{-\beta H^{(L,+)}_{\textrm{XXZ}}}   \ket{N^{(L)}_\lambda}} \,. 
\ee
For the state $\langle \cdot \rangle_-^{\textrm{(XXZ)}} $,  we reverse the sign in front of  $\lambda$ in \eqref{N_vector}, and apply the operator $H^{(L,-)}_{\textrm{XXZ}}$. 

Note that for fixed $ L \in 2 \mathbb{N} $, the limit $ \beta \to \infty $ in~\eqref{eq:101} converges to the finite-volume ground-state
of $ H^{(L,+)}_{\textrm{XXZ}} $, which is found in the subspace 
\begin{equation}
 S^z_{\textrm{tot}} := \sum_{u=-L+1}^L \tau^z_u/2 = 0 
 \end{equation} 
 where it is unique. 
This follows from a standard Perron-Frobenius argument, which is  enabled here by  the positivity and transitivity of the semigroup on that subspace. 
As a consequence, the finite-volume ground-state can be construction through the limit $ \beta \to \infty $ starting from any non-negative seed vector with $ S^z_{\textrm{tot}} = 0 $. The vectors $ N^{(L)}_\lambda $ with $ \lambda \in \mathbb{R} $ arbitrary are examples of such seed vectors and the limit~\eqref{eq:101} does not depend on the choice of $ \lambda $ in the seed (but still depends on $ \lambda $ through $  H^{(L,+)}_{\textrm{XXZ}} $.)

\medskip

Next we start the detailed discussion by recalling the probabilistic loop representations of the states described above.   The construction is included here mainly to keep the paper reasonably self-contained, since it is already contained in~\cite{AN}.

\section{Functional integral representation of the thermal states}\label{sec:FK}

\subsection{The general construction}

Thermal states of $ d $-dimensional quantum systems  can always be expressed in terms of a $(d+1)$-dimensional functional integral.  When the integrand can be expressed in positive terms, the result is a relation with a statistic-mechanical system in dimension $d+1$. 
General discussion 
of this theme and applications for specific purposes can be found e.g.\ in~\cite{Fey, Gin, AL, T, AN,Uel,BU}.
Our aim in this section is to  present this relation for the models discussed here.

As a starting point, let us note the following elementary identity, in which the power expansion of $e^{\beta K} $, which is valid for any bounded operator $K$, is cast in probabilistic terms:
\be \label{Poisson0}
e^{\beta (K -1)} = \sum_{n=0}^\infty  e^{-\beta} \int_{0< t_1< \dots< t_n < \beta} K_{t_N} \dots K_{t_1}   
 dt_1 \dots dt_n  = \int   \mathcal{T}  (\prod_{t \in \omega} K_t )  \, \, \rho_{[0,\beta]} (d \omega)  \,.
\ee 
In the last expression, the sequence of times is presented as a random point subset  $\omega =(t_1,\dots,t_n) \subset [0,\beta]$ distributed as a Poisson process on $[0,\beta]$ with intensity measure $dt$. The Poisson probability distribution is denoted here by $\rho_{[0,\beta]}(d\omega)$.  Attached to each point $t\in \omega$ is a copy of the operator $K$ labeled by $t$.   
The factors $K_t$ are rearranged  according to their  time label, which is denoted using the time ordering operator~$ \mathcal{T} $.
The  integral reproduces the familiar power series.   

For operators which are given by sums of (local) terms, as in our case 
\be \label{Ksum}
H_{\Lambda} = - \sum_{b\in \mathcal {E}(\Lambda)} K_b 
\ee 
with $K_b$  indexed by  the edge-set $\mathcal {E}(\Lambda) $ of a graph $\Lambda$,   the identity~\eqref{Poisson0} has the following extension
\begin{equation} \label{Poisson} 
e^{\beta  \sum_{b\in \mathcal {E}(\Lambda)} (K_b - 1)} \  = \ 
\int   K_{b_{|\omega|},t_{|\omega|}} \cdot \ldots \cdot K_{b_2,t_2} \cdot K_{b_1,t_1} \; \rho_{\Lambda\times[0,\beta]}(d\omega)   \, 
\end{equation} 
where $\omega$ are the configurations of a  Poisson point process over   $\mathcal {E}(\Lambda) \times [0,\beta]$,    
which may be depicted as collections of rungs of a random multicolumnar ladder net whose rungs are listed  as $\{(b_j,t_j)\}$ in increasing order of $t$.  
We denote by $\Omega_{\Lambda,\beta} $ the space of such configurations, and by    $\rho_{\Lambda\times[0,\beta]}(d\omega) $ the  
Poisson process with intensity measure $dt$ along the collection of  vertical columns $\cup_{b\in \mathcal {E}(\Lambda)} \{b\}\times [0,\beta] $.   

Given an orthonormal basis  $\{ \ket{\alpha} \}$ of the Hilbert space in which these operators operate,  one  has
\begin{eqnarray}  \label{time_dep}
\bra{\alpha'}  \mathcal{T} \Big( \prod_{(b,t) \in \omega} K_{b,t}\Big) \ket{\alpha} &=&  
\sum_{\widetilde \alpha}  \1[\omega, \widetilde \alpha]\,   \, 
\1\left[\substack{\alpha(t_{|\omega|})=\alpha'\\  \, \alpha(0)=\alpha}\right] \, W(\widetilde \alpha)  \\ 
W(\widetilde \alpha) &:=&  
 \prod_{j=1}^{|\omega|} \bra{\alpha(t_j+0) } \,  K_{b_j,t_j} \, \ket{\alpha(t_j-0) }  \notag
\end{eqnarray}  
where $\widetilde \alpha$ is summed over  functions $\widetilde \alpha :  [0,\beta] \mapsto \{ \ket \alpha\}$ which are constant between the transition times $0<t_1<...< t_{|\omega|}<\beta$, and the consistency constraint is expressed in the  indicator function $\1[\omega, \widetilde \alpha]$.   

Applying this representation, one gets
\begin{eqnarray} \label{Poisson2} 
\tr e^{\beta  \sum_{b\in \mathcal {E}(\Lambda)} (K_b - 1)}  & =&  
\int  \sum_{\widetilde \alpha: \, \alpha(\beta)=\alpha(0)}  \1[\omega, \widetilde \alpha]\,   \, 
  \, W(\widetilde\alpha)  \,\,
\rho_{\Lambda\times[0,\beta]}(d\omega).    
\end{eqnarray} 
The left side is obviously non-negative.  If a basis of vectors $\ket{\alpha}$ can be found in which also the matrix elements of $K_b$ are all non-negative, then \eqref{time_dep} yields a functional integral for the quantum partition function in which the integration is over $(\omega, \widetilde \alpha)$ which resembles a ``classical'' statistic mechanical system in $d+1$ dimensions (with $\alpha(t)$  a time-dependent configuration which changes at random times).

In that case one also gets a potentially  useful decomposition of the thermal state: 
\be \label{thermal}
\frac{\tr e^{-\beta H_\Lambda}  F } 
{\tr e^{-\beta H_\Lambda} }   = 
 \int  \mathbb{E}\left(F | \omega\right)  \mu_{\Lambda\times[0,\beta]}(d \omega) 
\ee 
with 
\begin{align}  \label{quasi_state}
\mathbb{E}\left(F | \omega\right)  & = 
 \tr \mathcal{T} \Big( F  \prod_{(b,t) \in \omega} K_{b,t}\Big) \big\slash  
  \tr \mathcal{T} \Big(  \prod_{(b,t) \in \omega} K_{b,t}\Big) 
\notag \\  
\mu_{\Lambda\times[0,\beta]}(d\omega)  & =   
\tr \Big[ \mathcal{T} \Big( \prod_{(b,t) \in \omega} K_{b,t}\Big) \Big] \; \rho_{\Lambda\times[0,\beta]}(d\omega)  \big \slash   \tr e^{-\beta (H_\Lambda +1)}  \, . 
\end{align} 
The functional $ F \mapsto \mathbb{E}\left(F | \omega\right) $ was dubbed in \cite{AN} a quasi-state. It does not possess the full positivity of a quantum state on all observables, but is a proper state on the sub-algebra of observables which are diagonal in the basis in which the interaction terms $K_b$ are all non-negative.  

A similar  decomposition is valid for states  
$ \bra {\Psi}   e^{-\beta H_\Lambda/2}  F e^{-\beta H_\Lambda/2}  \ket{\Psi} $,  
which are seeded by vectors~$ \Psi$ with non-negative overlaps with the above base vectors. 
For that it is pictorially convenient to cyclically shift the time interval to 
$[-\beta/2,\beta/2]$, and consider $\omega$ given by the  Poisson process over the set
\be 
\Lambda_{L,\beta}:=\Lambda_{L}\times[-\beta/2,\beta/2],
\ee
whose law is denoted by $\rho_{\Lambda_{L,\beta}}$.

Such non-negative functional integral representations of quantum states 
are associated with  Gibbs states of a classical statistic mechanical systems.   
Under this correspondence, non-uniqueness of the ground-states of a $d$-dimensional quantum spin system, in the infinite-volume limit, is  associated with a first-order phase transition (at a non-zero  temperature) of the corresponding $d+1$ dimensional classical system.   

\subsection{A potential-like extension} 

We shall also use an extension of the above expressions  to  operators of the form
\be \label{Ksum_U}
H_{\Lambda} = - \sum_{b\in \mathcal {E}(\Lambda)} K_b   -  V
\ee 
with $V$ an operator which is diagonal in the basis $\{\ket{\alpha}\}$, 
with $V\ket{\alpha} = V(\alpha) \ket{\alpha}$.   
In a manner reminiscent of the way that potential appears in the 
Feynman-Kac formula, one has 
\begin{equation} \label{Poisson3} 
\bra{ \alpha' } e^{\beta  \sum_{b\in \mathcal {E}(\Lambda)} (K_b - 1) + V}  \ket{\alpha} 
= \int  \sum_{\substack{ \widetilde \alpha: \, \alpha(0)=\alpha\\ 
\alpha(\beta) = \alpha'}}  \1[\omega, \widetilde \alpha]\,   \, 
  \, W(\alpha(0),\alpha(0))  \, e^{\int_{-\beta/2}^{\beta/2}    V(\alpha(t)) \, dt }\,
\rho_{\Lambda\times[0,\beta]}(d\omega).  \hspace{1cm} 
\end{equation} 
as can be deduced from  \eqref{time_dep}, e.g.~using  the Lie-Trotter product formula.

\section{Loop measures associated with $H_\textrm{AF}$} \label{sec:AFloop}

\subsection{The $H_\textrm{AF}$ seeded states} 

The positivity assumption does hold in the case of the two families of quantum spin chains considered here. Under the unitary (gauge) transformation $ U_L  := \exp\left( i \tfrac{\pi}{2} \sum_u (-1)^v S_u^z \right)  $,
the interaction terms of $H_{\textrm{AF}}^{(L)}$ acquire positive matrix elements in the standard basis of the joint eigenstates of~$(S^z_u)_{u\in \Lambda_L}$
\begin{equation}\label{constraints} 
U_L^* P_{uv}^{(0)} U_L = \frac{1}{2S+1} \sum_{m,m'=-S}^S \left| m, -m \rangle_{u,v} \langle m', -m' \right| \,  .
\end{equation}
In this basis, the factors $K_{b}= (2S+1) U_L^*  P_{uv}^{(0)}U_L$ which appear in~\eqref{Poisson} 
reduce to  constraints imposing the condition that before and after each rung the two spins at its edges add to zero.  To compute the global effect of that, one may replace each rung by a pair of ``infinitesimally separated'' lines, and then decompose the graph  into non-crossing loops, as indicated in Fig.~\ref{fig:loopsBC}.   

By elementary considerations~\cite{AN}, it follows that for each rung configuration $\omega$ drawn on $ \Lambda_{L,\beta}$:
 \be  \label{eq:ABloops}
\bra { D_{L} } \,  \mathcal{T} \Big( \prod_{(b,t) \in \omega} K_{b,t}\Big) \, \ket{ D_{L} }  
\ = \  (2S+1)^{N_\ell(\omega) }  \, , 
\ee  
where $N_\ell(\omega)$  is the number of loops into which the set of lines decomposes when the vertical lines are turned into columns through ``capping''  them at $t=\pm \beta/2$ over every other column starting with the left-most, cf.~Fig.~\ref{fig:loopsBC}.  
Depending on the parity of $ L $, the capping rule thus follows the two pairings  in Fig.~\ref{fig:pairing}. 

\begin{figure}[h]
\begin{center}
\includegraphics [width=0.5\textwidth]{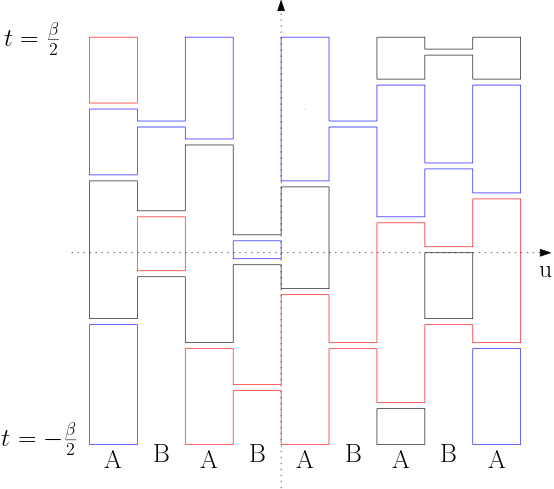}
\caption{
A configuration of randomly placed horizontal rungs in case $ L = 5 $, and its collection of loops obtained from the alternating boundary conditions at $ t = \pm \beta/2 $.       
Each rung imposes U-turns on the loops reaching it. }
\label{fig:loopsBC}
\end{center}
\end{figure}

More generally, for $ \ket{\pmb{m}} = \ket{m_{-L+1}, \dots, m_L} \in \mathcal{H}_L $ the orthonormal eigenfunctions of $ \{S_u^z\} _{u\in \Lambda_L} $, the matrix elements 
$ 
 \bra{\pmb{m}'}  
  \mathcal{T} \left( \prod_{(b,t) \in \omega} K_{b,t}\right) \ket{\pmb{m}}      
$ 
are given by the sum over configurations of the function $m: \Lambda_{L,\beta} \mapsto \{ -S, -S+1, \dots, S \} $ for which $m(x,t)$ is piecewise constant in time changing only at the encounters with the rungs of $\omega$, subject to the constraints explained next to \eqref{constraints}, and which at  $t=\pm\beta/2$  agree with $\ket{\pmb{m}}$ and   $\ket{\pmb{m}'}$ correspondingly.  

Adapting the quasi-state decomposition to the above seeded states, one gets:

\begin{proposition}[cf.~Prop.~2.1 in~\cite{AN}]  \label{prop_Dmeasures}
For the expectation value \eqref{eq:expdimer} corresponding to the seed vector $  \ket{D_L} $ and any observable $F$:

\begin{equation}
\left\langle F \right\rangle_{L,\beta}   = 
 \int  \mathbb{E}\left(F | \omega\right)  \mu_{L,\beta}(d \omega) ,
\end{equation}
where
\begin{equation} \label{quasi_D}
 \mathbb{E}\left(F | \omega\right) := \frac{1}{(2S+1)^{N_\ell(\omega) }}  \left\langle D_{L} \right| \mathcal{T} \Big( \prod_{\substack {(b,t) \in \omega \\ t \in [0,\beta/2) }}  K_{b,t}\Big) \, F \;  \mathcal{T} \Big( \prod_{\substack {(b,t) \in \omega \\ t  \in [-\beta/2, 0) }}  K_{b,t}\Big)  \left| D_{L} \right\rangle \, 
\end{equation}
and 
\be \label{tilt}
\mu_{L,\beta}(d\omega) :=   
\frac{1}{\textrm{Norm.}} \sqrt{Q}^{N_\ell(\omega) }\,  \rho_{\Lambda_{L,\beta}}(d\omega)  
  \qquad   \mbox{at $
 \sqrt Q = 2S+1 $}\, .
\ee 
\end{proposition}  

Behind the complicated looking formula \eqref{quasi_D} is a  simple rule which  is particularly easy to describe for observables $F$ which are functions of the spins $S^z_u$.   
The conditional expectation conditioned on $\omega$ is obtained by averaging the value of $F$ over spin configurations which vary independently between the loops of $\omega$. On each loop the spins are constrained to assume only two values, changing the sign upon each U-turn.  

Following are some instructive examples: 
\begin{enumerate}
\item 
For each $\omega$ 
\be 
\mathbb{E}\left(S^x_u S^x_v | \omega\right) =  
\mathbb{E}\left(S^z_u S^z_v | \omega\right) = 
(-1)^{u-v} \, C_S\, \,
\1[(u,0) \stackrel {\omega}{\leftrightarrow}  (v,0)] 
\ee 
where $C_S= \sum_{m=-S}^S m^2 / (2S+1)$ 
and the space-time points $(u,0) \stackrel {\omega}{\leftrightarrow}  (v,0)$ denotes the condition that 
$(u,0)$ and $(v,0)$ lie on the same loop of $\omega$.
\item 
For the projection operator defined by \eqref{eq:pro}  
\begin{eqnarray}  \label{eq:proexp}
\mathbb{E}[(2S+1) P_{u,v}^{(0)} | \omega] &=& 
\begin{cases} 
\,\,\, \,\, 1 & \mbox{if  $(u,0) \stackrel {\omega}{\leftrightarrow}  (v,0)$ } \\   
(2S+1)^{-1}  & \mbox{if  not} 
\end{cases}    \notag \\[2ex]  
&=&  \Big( 1+ 2S \ \1[\,(u,0) \stackrel {\omega}{\leftrightarrow}  (v,0)\, ] \Big)/ (2S+1)\,.  
\end{eqnarray}  
\end{enumerate} 

 \subsection{The $H_{AF}$ thermal equilibrium states}  
 
 The above representation has a natural extension  to the thermal Gibbs states, for which the expectation value functional is given by
\be \label{Gibbs}  
\frac{\tr F e^{-\beta H_{\textrm{AF}}^{(L)}}} {\tr e^{-\beta H_\textrm{AF}^{(L)}} } \, .
\ee 
In this case the above construction yields a representation in terms of random loop decomposition of $\Lambda_{L,\beta}$ constructed with the time-periodic boundary conditions, with loops continuing directly from $t=\pm\beta/2$.   And if the  quantum Hamiltonian $ H_{AF}^{(L)} $ is taken with periodic boundary conditions then also the spacial coordinate is periodic, i.e.~the loops are over a torus.  
Similarly as in~\eqref{eq:ABloops} one gets
\be \label{eq:traceloops}
\tr \mathcal{T} \Big( \prod_{(b,t) \in \omega} K_{b,t}\Big) =  \ (2S+1)^{N_\ell^{\textrm{per}}(\omega)}\,,   
\ee 
where $N_\ell^\textrm{per}(\omega)$ is the number of loops into which the set of lines decomposes with the time-periodic boundary condition under which $t=\pm\beta/2$ are identified.

With this adjustment in the assignment of loops  to rung configurations, the state's  representation in terms of the loop system with the probability distribution \eqref{tilt}  remains valid also in the presence of  periodicity of either the temporal or spacial direction. 
This point should be borne in mind in the discussion which follows. In the pseudo spin representation, which is described next, a distinction will appear between the weights of winding versus contractible loops.

From~\eqref{eq:traceloops} and~\eqref{eq:ABloops}  we also obtain the  following explicit justification for~\eqref{gsproj}:  
\begin{equation}\label{eq:conditionAF}
\frac{\langle D_{L} | e^{-\beta H_{\textrm{AF}}^{(L)}} |  D_{L} \rangle }{\tr e^{-\beta H_{\textrm{AF}}^{(L)}} } = \frac{\int \sqrt{Q}^{N_\ell(\omega)} \rho_{\Lambda_{L,\beta}}(d\omega)  }{\int  \sqrt{Q}^{N_\ell^{\textrm{per}}(\omega)} \rho_{\Lambda_{L,\beta}}(d\omega)  } \geq \frac{1}{\sqrt{Q}^L} \, . 
\end{equation}
Indeed, for fixed rung configuration $ \omega $, the loops in the denominator are constructed on the time-periodic version of  $\Lambda_{L,\beta}$ and the loops in the numerator arise in the capped version of $\Lambda_{L,\beta}$. 
 Since the addition of a rung changes the number of loops by $ \pm 1 $ (depending on whether the two points were already connected by a loop or not), we have $ | N_\ell(\omega) - N_\ell^\textrm{per}(\omega) | \leq L $ and hence the lower bound in~\eqref{eq:conditionAF} follows.

\section{The loop representation of the anisotropic XXZ-model}  

\subsection{A modified $4$-edge presentation of the $XXZ$ interaction}

We shall now show that the loop measure which appeared quite naturally in the representation of the ground-states of $H^{(L)}_\textrm{AF}$ plays a similar role also for the $H^{(L)}_\textrm{XXZ}$ spin system.  
Preparing for that, we rewrite the Hamiltonian of the XXZ chain in terms of the slightly modified local interactions  consisting of the sum of the following four rank-one operators
\begin{align}\label{eq:repK} 
 K_{v,v+1} =     \Big(| -, + \rangle    \langle +, - |  \  +\   | +, - \rangle  \langle  -, + | \  +\   e^{\lambda} \,  |  -, +\rangle \langle  -,+ | \  +\    e^{-\lambda} \,  | +, - \rangle   \langle  +, -  | \Big)_{v,v+1}   \, .
\end{align}
(written in  the  bra-ket notation of \eqref{ket_bra}, with $ | \pm , \pm \rangle $ the eigenfunctions of $(\tau^z_v, \tau^z_{v+1})$).

The action of  $K_{v,v+1}$ is depicted in Fig.~\ref{fig:4edges}   in terms of  the four edge configurations with the weights: 
\be \label{eq:weights}
W_a = 1\,, \quad W_b = 1\,, \quad W_c = e^{-\lambda}\,, \quad W_d = e^{\lambda}\,.
\ee 
In this  representation of $H^{(L)}_{\textrm{XXZ}}$, the local interaction  terms are no longer invariant under spacial reflection, 
but their sum differs from the more symmetric expression \eqref{H_XXZ} only in a boundary term -- in fact the one which was included in \eqref{XXZ_bc}  due to  this correspondence.  
Furthermore, this boundary term does not appear in the operators' periodic version
\begin{equation}\label{XXZper}
H^{(L,\textrm{per})}_{\textrm{XXZ}} :=  - \frac{1}{2} \sum_{v=-L+1}^{L}\,  
\left(  \tau^x_{v}  \tau^x_{v+1}+  \tau^y_{v}  \tau^y_{v+1} + \cosh(\lambda)\left( 1-  \tau^z_{v}  \tau^z_{v+1} \right) \right) \, ,
\end{equation}
where the sum extends also to the edge connecting $ L $ and $-L+1 \equiv L+1 $. Following is the exact statement. 
\begin{lem}\label{lem:XXZeqv}
For any $ L \in \mathbb{N} $   and $ \lambda \in \mathbb{R}   $: 
\begin{equation}\label{eq:XXZequiv}
H^{(L)}_{\textrm{XXZ}} +  \sinh(\lambda)\ \frac{\tau_{-L+1}^z- \tau_L^z}{2}  \ = \ -   \sum_{v=-L+1}^{L-1}\,  K_{v,v+1}   \quad  ( \ :=  K^{(L)} \ )
 \end{equation}
Furthermore, taken with the periodic boundary conditions the two operators agree without the boundary term: 
\be \label{eq:H_K_per}
 H^{(L, \textrm{per})}_{\textrm{XXZ}}   \ =   \  -   \sum_{v=-L+1}^{L-1}\,  K_{v,v+1}   - K_{L, -L+1} \quad  ( \ :=    K^{(L,\textrm{per})}  \ ) \, . 
 \ee
\end{lem}

\begin{figure}[h]
\begin{center} 
 \includegraphics [width=0.3\textwidth]{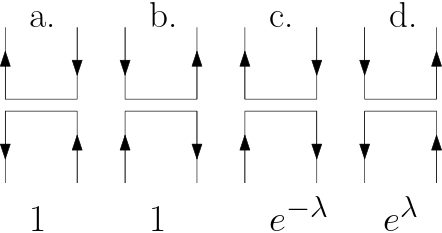} 
\caption{The non-zero matrix elements of the two-spin operator $K_{v,v+1}$ of \eqref{eq:repK}, with up arrows corresponding to $ \tau = 1 $ and down arrows to $ \tau = -1 $. These weights can be reinterpreted as the product of $e^{\lambda/2}$ per left U-turn 
and $e^{-\lambda/2}$ per right U-turn along the $\tau$-oriented loop lines. }
\label{fig:4edges}
 \end{center}
\end{figure}

\begin{proof}
 The action of the sum of the first two 
edges (a.~and~b.) agrees with that of $\left[ \tau^x_{v}  \tau^x_{v+1}+  \tau^y_{v}  \tau^y_{v+1}\right]/2  $, which represent the local $ x $- and $ y $-terms in~\eqref{H_XXZ}.  The local $ z $-terms in~\eqref{H_XXZ} and~\eqref{XXZper} agree with the action of the last two edges (c.~and d.) in Fig.~\ref{fig:4edges}. However, their weight in~\eqref{H_XXZ}  and~\eqref{XXZper} is $ \cosh(\lambda)$  for both edges c.~and~d.. 
The fact that the summation over all edges in the non-periodic box $ \Lambda_L $ yields the same result up to a boundary term is checked by noting that for a given spin configuration $ \pmb{\tau}  $ the difference between these two cases can be expressed in terms of the number of up- and down-turns, $ n_\uparrow^{(L)}(\pmb{\tau})$, $ n_\downarrow^{(L)}(\pmb{\tau}) $, over the edges of $ \Lambda_L$:
\begin{align}
 K^{(L)}  + H^{(L)}_{\textrm{XXZ} } \ & =    \     e^\lambda n_\uparrow^{(L)}(\pmb{\tau}) + e^{-\lambda} n_\downarrow^{(L)}(\pmb{\tau})   + \sum_{v=-L+1}^{L-1} \cosh(\lambda)  \frac{ \tau^z_{v}  \tau^z_{v+1}  -1}{2}   \notag \\
 & = \  e^\lambda n_\uparrow^{(L)}(\pmb{\tau}) + e^{-\lambda} n_\downarrow^{(L)}(\pmb{\tau})  - \cosh(\lambda)  \big( n_\uparrow^{(L)}(\pmb{\tau}) + n_\downarrow^{(L)}(\pmb{\tau})  \big)  \notag \\
 & = \   \sinh(\lambda)  \big( n_\uparrow^{(L)}(\pmb{\tau}) - n_\downarrow^{(L)}(\pmb{\tau})  \big)  \, . 
 \end{align}
The proof of~\eqref{eq:XXZequiv} is completed by noting that $ n_\uparrow^{(L)}(\pmb{\tau}) - n_\downarrow^{(L)}(\pmb{\tau})  =  (\tau_L^z - \tau_{-L+1}^z) /2 $. In the periodic case, this boundary term drops out.
\end{proof}

\subsection{A link between the $H_\textrm{XXZ}$ and  $H_\textrm{AF}$ loop measures} 

Applying the general procedure  to the operator $e^{-\beta H^{(L,+)}_{\textrm{XXZ}}/2}$ written as  $e^{\beta K^{(L)}/2}$ we obtain a representation of states in terms a functional integral over configurations $\vec\omega = (\omega,\tau)$ with binary-valued functions
$$
\tau: \Lambda_{L,\beta} \to \{ -1, 1 \} 
$$
whose values may change only at the rungs of $\omega$, consistently with the edges depicted in Fig.~\ref{fig:4edges}.  
The local condition implies that the allowed functions $\tau$ are consistent with the loop structure of $\omega$:  Along each loop of $\omega$ the function $\tau$  is aligned with  either its clockwise of counterclockwise orientation.   
We  denote by $\1[\omega,\tau] $  the indicator function expressing this consistency condition.

\begin{theorem}\label{thm:FKXXZ} 
For $ \lambda \geq 0 $, any $L$ even and $\beta$, the expectation value of any function of $\tau^z $ in the state defined in~\eqref{eq:101} is given by 
\be
\left\langle f(\tau^z) \right\rangle_{L,\beta,\pm\lambda}^{(\textrm{XXZ},\pm)}   =  \int  \mathbb{E}_\pm\left(f | \omega\right)  \mu_{L,\beta}(d \omega) 
\end{equation}
with $\mu_{L,\beta}$ the measure defined in \eqref{tilt} at 
\be\label{lambda_Q}
\sqrt Q = e^\lambda +e^{-\lambda}
\ee
and the normalized expectation value
\be \label{cond_tau}
\mathbb{E}_\pm\left(f | \omega\right)  = \frac{1}  {\sqrt Q^{N_\ell(\omega)}} \sum_\tau \1[\omega,\tau] \ W_\pm(\omega,\tau)  \ f(\tau(\cdot, 0))
\ee
with the weights  
\begin{eqnarray}  
W_\pm(\omega,\tau) :=  
\Big(\prod_{ + \ell } e^{\pm\lambda} \Big)  \, \Big( \prod_{ -\ell  }e^{\mp\lambda}   \Big)  \, , 
\end{eqnarray}
where the product is over $ (+) $ and $ (-) $ oriented loops $ \ell $ of $ (\omega,\tau) $.
\end{theorem} 
\begin{proof}   We spell the proof in the case $+ $. Proceeding as described in Section~\ref{sec:FK}, we get 
\be \label{eq:expandXXZ}
\frac{\bra{N^{(L)}_+}    e^{-\beta H^{(L,+)}_{_{XXZ}}/2} \,   f(\tau^z)\,   e^{-\beta H^{(L,+)}_{_{XXZ}}/2}   \ket{N^{(L)}_+}} 
{\bra{N^{(L)}_+}    e^{-\beta H^{(L,+)}_{_{XXZ}}}   \ket{N^{(L)}_+}} = 
\frac{ 
\int \sum_\tau \1[\omega,\tau] \widetilde W_+(\omega,\tau)   f(\tau(\cdot, 0))  \rho_{\Lambda_{L,\beta}}(d\omega) }
{\int  \sum_\tau \1[\omega,\tau] \widetilde  W_+(\omega,\tau)   \rho_{\Lambda_{L,\beta}}(d\omega)   } 
\ee 
with   
weights   given by the product over all rungs of $ \omega $ in terms the four types $ \#(\tau,b) \in \{ a., b., c., d.\} $  listed in~\eqref{eq:weights} (cf.~Fig.~\ref{fig:4edges}):
\begin{eqnarray}  
\widetilde W_+(\omega,\tau) &=&    \prod_{b\in \omega } W_{\#(\tau,b)}   .\end{eqnarray}
Lumping the factors by the loops of $\omega$, for each loop which does not reach the upper and lower boundary of the box $ \Lambda_{L,\beta}$, one gets the total of $e^{+ \lambda} $ per   counter-clockwise ($+$)  and $e^{- \lambda} $ per  clockwise ($-$) oriented loop.  In that case $\widetilde W_+(\omega,\tau)$ reduces to the above defined $W_+(\omega,\tau) $.   Furthermore, 
with our choice of the seed vector $\ket{N^{(L)}_+}$  that  is also true  of the loops which are reflected from the upper and/or the lower boundary.  

Summing over the $2^{N_\ell(\omega)}$  possible loop orientations one gets, for each $\omega$
\be \label{eq:347}
\sum_\tau \1[\omega,\tau] \ W_\pm(\omega,\tau)  = (e^\lambda + e^{-\lambda})^{N_\ell(\omega)} =  \sqrt Q^{N_\ell(\omega)}\,. 
\ee 
Thus, the average in \eqref{eq:expandXXZ} is over $(\omega,\tau)$ with the joint distribution  whose marginal distribution of $\omega$ is the normalized probability measure
\be
\frac{1}{\text{Norm.}}  (e^\lambda + e^{-\lambda})^{N_\ell(\omega)} \rho_{\Lambda_{L,\beta}}(d\omega) = \mu_{L,\beta}(d\omega) \,,  
\ee
with the conditional distribution of $\tau$ conditioned on $\omega$ stated in \eqref{cond_tau}. 
\end{proof}

It may be instructive to pause here and compare the different perspectives on the 
above loop measure. 
Starting from the analysis of the two different quantum spin chains we arrive at a common system of random rung configurations $\omega$, whose probability distribution in both models takes the form 
\be \label{eq:tilted} 
\mu_{L,\beta}(d\omega) = \sqrt Q^{N_\ell(\omega)} \rho_{\Lambda_{L,\beta}}(d\omega) /\text{Norm.} \quad 
\mbox{at \quad $2S+1 \,  = \, \sqrt Q \,  = \,  e^\lambda +e^{-\lambda}  $} 
\ee  
with $\rho_{\Lambda_{L,\beta}}(d\omega)$ a Poisson measure of intensity one.  The  factor $\sqrt Q^{N_\ell(\omega)} $, by which the measure is tilted, appears  through the summation over another degree of freedom, at which point the models differ.  
More explicitly, in the different systems this common factor is  variably decomposed as
\begin{eqnarray}  
\sqrt Q^{N_\ell(\omega)} &=& \sum_{m} \1[\omega,m]    \hspace{3.5cm}  (H_\textrm{AF})  \notag   \\ 
&=& \sum_{\tau} \1[\omega,\tau] \quad   \prod_{b\in \omega } W_{\#(\tau,b)}   \hspace{1cm}  \quad (H_\textrm{XXZ})
\end{eqnarray}
where the summations are over functions 
\begin{eqnarray}  m:  \Lambda_L &\mapsto& \{-S,-S+1,\dots,S\} 
\quad \mbox{ (with $2S+1= \sqrt Q$)} ,\notag \\ 
\tau:   \Lambda_L &\mapsto&  \{-1,+1\}  \hspace{2.1cm} \mbox{ (with $e^\lambda + e^{-\lambda} = \sqrt Q$) } .
\end{eqnarray} 
The indicator functions impose the consistency condition requiring $m$ or $\tau$  to be consistent with the loop structure of $\omega$, i.e.~a switch of signs at each U-turn and otherwise be constant along each vertical segment.

Thus,  the above system of the random \emph{oriented loop}  described by $\vec{\omega} = (\omega,\tau)$  can be presented in two equivalent forms: 
\begin{enumerate}
\item
Locally: as a $4$-edge model of random oriented lines with the weights listed in Fig.~\ref{fig:4edges}.  
\item 
Globally: by the following two characteristics of  its probability distribution  $ \widehat\mu_{L,\beta,\lambda} $:
\begin{enumerate} 
\item[i)]  $\omega$ has the probability distribution $\mu_{L,\beta}$  which is tilted relative to the Poisson process $\rho_{\Lambda_{L,\beta}}(d\omega)$ by the factor $\sqrt Q^{N_\ell(\omega)}$
\item[ii)] conditioned on $\omega$, the conditional distribution of $\tau$ corresponds to   independent assignments of orientation to the loops of $\omega$, at probabilities $
e^{\pm \lambda}/[e^{\lambda}+e^{-\lambda}]$ depending on whether the loop is anticlockwise ($+$) or cklockwise ($-$) oriented.
\end{enumerate} 
To emphasise the fact that the measure $ \widehat\mu_{L,\beta, \lambda} $ changes under a change of the sign of $ \lambda \in \mathbb{R} $,
we keep track of it in the notation.
\end{enumerate} 

The above  local to global relation  is  reminiscent of the Baxter-Kelland-Wu~\cite{BKW76}   correspondence between the  $Q$-state Potts model  and the  $6$-vertex model, which followed the analysis of Temperley and Lieb~\cite{TL71}.

In the context of the XXZ-operator, the loop  picture carries a particularly simple implication for sites at the boundary of $\Lambda_L$, where the relation of $\tau(u,t)$ to loop's helicity is unambiguous.  One gets, for the finite volume ground-states: 
\be 
\langle \tau(u,0) \rangle_{L,\beta} = \begin{cases} 
- \tanh (\lambda) & u = -L+1, \\ 
+ \tanh (\lambda) & u = L,
\end{cases} \,  
\ee 
regardless of the value of $ L $ and $ \beta>0$.     

\subsection{The XXZ-Hamiltonian with the periodic boundary conditions}

Under  the joint distribution  $ \widehat\mu_{L,\beta,  \lambda} $ on oriented loops, 
the induced measure on $ \tau $'s restriction to the line $ t = 0 $   
 was shown to agree with the 
seeded expectation value of $z$-spins for the XXZ-Hamiltonian with boundary term on $ \Lambda_L $.

This relation takes a simpler form for the thermal state of 
the XXZ-Hamiltonian taken with periodic boundary conditions~\eqref{XXZper}.  
To express that, we denote by $\widehat\mu_{L,\beta,  \lambda}^\textrm{per}$  the similarly defined the 4-edges measure on $\Lambda_{L,\beta}$, taken with periodic boundary conditions in both space and time direction. 

\begin{theorem}  
The marginal distribution of $ \widehat\mu^\textrm{per}_{L,\beta, \lambda} $ on orientations $ \tau $ coincides with the quantum expectation of the XXZ-model's tracial state, i.e.~for any finite collection of space-time points $ (u_j,t_j) $ which are ordered $ t_1 < t_2 < \dots < t_N $:
\begin{multline}\label{eq:XXZper}
 \frac{\tr\left( e^{-(\beta -t_N) H^{(L,\textrm{per})}_{\textrm{XXZ}}} P_{u_N}(\sigma_N) e^{- (t_N -t_{N-1}) H^{(L,\textrm{per})}_{\textrm{XXZ}}} \cdots P_{u_1}(\sigma_1)  \ e^{-t_1 H^{(L,\textrm{per})}_{\textrm{XXZ}}}   \right)}    
{\tr\left( e^{-\beta H^{(L,\textrm{per})}_{\textrm{XXZ}}} \right) } = \\ 
= \ \int  \ \prod_{j=1}^N 1\left[\tau(u_j,t_j) = \sigma_j \right]  \,  \widehat\mu^\textrm{per}_{L,\beta, \lambda}(d\vec\omega)     \, .
\end{multline}
where $ \sigma_j \in \{-1,1\} $ are prescribed spin values and $ P_u(\sigma) := 1 \left[ \tau^z_{u}  = \sigma\right]  $ stands for the projection operator onto states with $ \sigma $ as the $ z $-component of the spin at $ u $. 
\end{theorem} 
\begin{proof} 
The proof proceeds  by plugging the operator $  K^{(L,\textrm{per})}  $ from~\eqref{eq:H_K_per} into the loop representation~\eqref{Poisson} for 
each of the factors $ \exp[(t_j-t_{j-1} ) H^{(L,\textrm{per})}_{\textrm{XXZ}}]$ in the time-ordered product in the numerator.  The operator $ K^{(L,\textrm{per})} $ produces exactly the weights of the 4-edges model with spatially periodic boundary conditions. 
The projection operators $  P_{u_j}(\sigma_j)$ inserted behind each factor $\exp[(t_j-t_{j-1})  K^{(L,\textrm{per})}] $ 
fixes the spin-value to $ \sigma_j $ at the particular instance $(u_j,t_j)  $ in space-time. Evaluating the trace in the joint eigenbasis 
of $ \tau^z_u $ will enforce periodic boundary conditions of the oriented loops also in the time direction.
\end{proof}

Since the right-side in~\eqref{eq:XXZper} depends on $ \lambda $ only through the anisotropy parameter $ \cosh(\lambda) $ entering 
the periodic XXZ-Hamiltonian, the distribution of the pseudo-spins is easily seen to exhibit the following symmetry, which will play a crucial role in our proof of dimerization (Theorem~\ref{thm:main1}).
\begin{cor}\label{cor:symm}
Under $ \widehat\mu^\textrm{per}_{L,\beta,  \lambda} $, the marginal distribution  of  $ \tau $   is a symmetric function of~$ \lambda $. 
\end{cor}

\subsection{Further symmetry considerations} 

As a preparatory step towards the proof of N\'eel order,  let us discuss the symmetries of the oriented loop's distribution.  We start by denoting three mappings on the space of functions $\tau(u,t)$ which are defined by
\begin{eqnarray}  
\mathcal S[\tau](u,t) &=\tau(u-1,t)   & \mbox{one-step shift}  \notag \\ 
F[\tau](u,t) &= -\tau(u,t)   & \mbox{spin flip} \\ 
R[\tau](u,t) &=  \tau(-u+1,-t)  &\mbox{space$\times$time reflection w.r.t. $(1/2,0)$} \notag  
\end{eqnarray}  
and extend the last two to a similarly defined action on the un-oriented edge configuration $\omega$. 
 
The following is a simple, but very helpful observation.  
 
\begin{theorem} 
For each finite $L$, $\beta$,  and $\lambda$, the above joint probability distribution of $(\omega,\tau)$    is invariant under  $ R\circ F$.  
Furthermore, in any accumulation point  of such measures (e.g.~limit $\L\to \infty$  with $L$ of a fixed parity)  which  is  invariant under the two-step shift ($\mathcal S^2$),  the magnetization  satisfies
 \be   \label{alt}
\langle \tau(u,0) \rangle_{(L,\beta)}    =  (-1)^u  M 
\ee
for some $M\in [-1,1]$.\end{theorem}
To avoid confusion let us stress that 
\eqref{alt} does not yet establish the existence of N\'eel order.  For that, one needs to show that $M\neq 0$. 

\begin{proof}
The first statement follows readily from the above characterization  (i)-(ii) of the measure, as  under reflections the distribution of $\omega$ is invariant, but the loop's orientational preference is inverted.   

To prove the second statement we combine the above symmetry with the assumed two-step shift invariance. These imply
\be
\tau(2,0) \ = \ [ \mathcal S^2 \circ (R\circ F \tau)] (2,0)\  = \ (R\circ F \tau)] (0,0) = - \tau(1,0) \, .
\ee  
The full oscillation \eqref{alt} follows by another application of invariance under the double shift $\mathcal S^2$.  
\end{proof}

\section{The quantum loops system's critical percolation  structure}

\subsection{An FKG-type structure} 
The probability distribution  \eqref{tilt}  is reminiscent of the loop representation of the planar $Q$-state random-cluster  models. For details on the random-cluster model itself, we refer to the monograph \cite{Grim} and the lecture notes \cite{Dum17} (for recent developments). 

As in that case, it is relevant  to recognize here the presence  of a self-dual $A$/$B$-percolation model. 
To formulate it,  we partition any rectangle $\Lambda_{L,\beta} \subset \mathbb{Z} \times \mathbb{R} $ into a union of vertical columns of width $1$ over the edges of $ \Lambda_L $, labeled alternatively as $A$ and $B$,
\be 
A := \{(2n, 2n+1)\}_{n\in \Z}\,, \quad B :=\{(2n-1,2n)\}_{n\in \Z} \, , 
\ee 
with the column over $(0,1)$ marked as $A$.  
Rungs $ \omega $ are then distributed in the edge columns with respect to the probability measure $\mu_{L,\beta}$.

These rungs serve a dual role.  We interpret each  as  a cut in the column over which it lies and at the same time a bridge linking the two  domains which  are touched by its endpoints.  
\begin{figure}[h]
\begin{center}\begin{minipage}{6in}
  \centering
  \hspace*{.5in}
  $\vcenter{\hbox{\includegraphics[height=3 in]{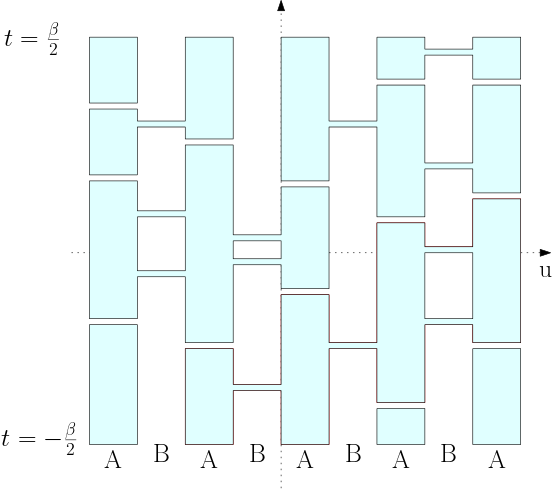}}}$
  \hspace*{.2in}
  $\vcenter{\hbox{\includegraphics[height=0.7 in]{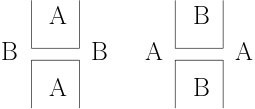}}}$
\end{minipage}
\caption{ The rungs of $ \omega $ both connect and, in the dual sense, disconnect: those placed over the   $A$ strips  (shaded in this picture) decrease the $A$ connectivity; those over $B$ strips increase it, and vice versa.   Since the complement of $ \Lambda_{L,\beta} $ is $ B $-connected this picture corresponds to $ B $-wired boundary conditions.} \label{fig:ABpicture2}\end{center}
\end{figure}
To visualize the $A$- and $B$-connected components, also called $A$- and $B$-{\em clusters}, which result from this convention it is convenient to think of each rung as having a small (infinitesimal) width and being bounded by a pair of segments,  as is indicated in Fig.~\ref{fig:ABpicture2}.     

Thus,  associated  with each configuration $ \omega$ is a decomposition of 
$\Lambda_{L,\beta} $ into $A$-clusters and $B$-clusters, with $A$-clusters bounded by $B$-clusters, and vice versa.   In the topological sense  this percolation model is self dual.   Also, the probability distribution is symmetric, except possibly  for asymmetry introduced by boundary conditions. As is explained below, this implies that the percolation model is at its phase transition point.   The transition can be continuous, as is the case for independent percolation ($Q=1$), or discontinuous as in models with $Q$ large enough.  This distinction is tied in with the existence or not of symmetry breaking in the ground-states of the two quantum models discussed here.   

The similarity with the random-cluster measures led \cite{AN} to introduce a partial order ($ {\prec}$) on the space of rung configurations  in which  $A$-connection is monotone increasing and $B$-connection is monotone decreasing.  More explicitly, labeling the rungs as of $ A $- or $ B $-type: $\omega_1\ \prec\ \omega_2$ if the $A$-connections in $\omega_1$ are all holding in $\omega_2$. 
This notion is useful since the measures~$\mu_{L,\beta}$
satisfy the Fortuin-Kasteleyn-Ginibre (FKG) lattice condition  which enables powerful monotonicity arguments.  The FKG structure was used in the proof of the AN-dichotomy \cite{AN} stated in Proposition~\ref{thm:dichotomy}. Here, we will use the following facts. First, it implies  a FKG inequality stating, for every events $E$ and $F$ that are {\em increasing} (meaning that their indicator functions are increasing for $\prec$):
\be
\mu_{L,\beta}[E\cap F]\ge \mu_{L,\beta}[E]\mu_{L,\beta}[ F].
\ee
Another implication of the FKG lattice condition is the monotonicity in so-called boundary conditions. Here, the boundary conditions are imposed by the structure of the underlying graph $\Lambda_{L,\beta}$, so we wish to draw a comparison with the random-cluster model. The construction with rungs at the top and bottom capping the loops implies that when $L$ is odd (as in Figure~\ref{fig:ABpicture2}), the  complement of the box $\Lambda_{L,\beta}$ is treated as $B$-connected, while when $L$ is even it is $A$-connected. Borrowing the language of the random-cluster model, we see that our capping procedure used in the construction of $\mu_{L,\beta}$ can be understood as enforcing $B$-wired or $A$-wired boundary conditions depending on the parity of $L$. To stress the type of the boundary condition and to draw an even more direct link to the standard theory of random-cluster models, in this section we write $\mu_{L,\beta}^\#$ instead of $\mu_{L,\beta}$, with $\#=A$ if $L$ is even, and $\#=B$ if $L$ is odd.
 
Now, consider $L\ge\ell$ with $\ell$ even and $\beta\ge t$. The measure $\mu_{\ell,t}^A$ can be seen as the measure $\mu_{L,\beta}^\#$ (with $\#$ equal to $A$ or $B$ depending on $L$ even or odd, or equal to $\mathrm{per}$ if one wishes) in which we place the so-called \emph{$A$-cutter}, since either the points in $\Lambda_{\ell,t}$ were already $A$-connected within  $\Lambda_{\ell,t}$ or, in case their $A$-connection ran through the complement of  $\Lambda_{\ell,t}$, this will still be true due to the fact that the boundary conditions render this complement into a single $A$-cluster. 
The monotonicity in boundary conditions therefore implies that  for an increasing event $E$ depending on rungs in $\Lambda_{\ell,t}$ only,
\begin{equation}\label{eq:Adom}
 \mu_{L,\beta}^\#(E) \leq \mu_{\ell,t}^A(E) \,.
\end{equation}
Likewise, if $\ell$ is odd and one uses a \emph{$B$-cutter} to cut out a smaller box, one gets
\begin{equation}\label{eq:Bdom}
 \mu_{\ell,t}^B(E) \leq \mu_{L,\beta}^\#(E) . 
\end{equation}

\subsection{Results based on the percolation analysis} 

By the monotonicity in the domain, the above FKG structure implies the convergence of the extremal measures, i.e.\ along increasing sequences of $A$- or $B$-wired boundary conditions:
\begin{align}\label{eq:evenoddlimit}
& \mu^A := \lim_{\substack{L\to \infty \\ L \textrm{\ even}}} \lim_{\beta \to \infty}  \mu^A_{L,\beta}  \notag \\
& \mu^B :=  \lim_{\substack{L\to \infty \\ L \textrm{\ odd}}}   \lim_{\beta \to \infty}  \mu^B_{L,\beta} \, , 
\end{align}
in the weak sense of convergence of probability measures on the configuration spaces of rungs on $\mathbb{Z} \times \mathbb{R} $.

To present the full resolution of the question posed by the dichotomy, we start with the following preparatory statements.  

\begin{theorem}\label{thm:mainpercolation}    
For any $Q\geq 1$, and regardless of whether the infinite-volume loop measures $\mu^A$, $\mu^B$ coincide:
\begin{enumerate} 
 \item Each of these measures is supported on configurations with only closed loops, 
 i.e.~there are no infinite boundary lines.   
\item  The convergence extends to that of the joint distribution of $(\omega,\tau)$, i.e.\ of  the ordered loop lines. 
\item The limiting measures'   conditional distribution of $\tau$, conditioned on $\omega$, is given by the same rule as in finite volume: at given $\omega$ the loops are oriented independently of each other with probabilities $e^{\pm \lambda}/[e^\lambda + e^{-\lambda}] $, at $\lambda$ satisfying $ \sqrt{Q}= e^\lambda + e^{-\lambda}$,  with $(-)$ for clockwise and $(+)$ for counter-clockwise orientation.  
\end{enumerate} 
In case the measures coincide ($\mu^A = \mu^B$), then
\begin{enumerate} 
\setcounter{enumi}{3}
\item  the limiting state is supported on configurations in which there is no infinite $A$-clusters or $B$-clusters, and instead each point is surrounded by an infinite family of nested  loops;
\item  the  loop measures with the periodic boundary conditions in both temporal and spacial direction $  \mu_{L,\beta}^\textrm{per}  $ converge to the shared limit as $ L, \beta \to \infty $.
\end{enumerate} 

\end{theorem}

The proof of this theorem will follow standard arguments in percolation theory that must still be adapted to the current context.

\subsection{Proofs}

We begin with three statements that will play important roles. The first one deals with ergodic properties of $\mu^A$ and $\mu^B$. Let $\mathcal{S}_x$ be a translation by $x\in \mathbb{Z}\times\mathbb R$. This translation induces a shift $\mathcal{S}_x\omega$ and $\mathcal{S}_xE$ of a configuration and an event.
Furthermore, an event $E$ is {\em invariant under translations} if for any $x\in \mathbb Z\times\mathbb R$, $\mathcal{S}_xE=E$. A measure $\mu$ is {\em invariant under translations} if $\mu[\mathcal{S}_xE]=\mu[E]$ for any event $E$ and any $x\in \mathbb Z\times\mathbb R$. The measure is {\em ergodic} if any event invariant under translation has probability 0 or 1.

\begin{lem}\label{lem:ergodicity}
The measures $\mu^A$ and $\mu^B$ are invariant under spacial translations by $ 2\mathbb{Z} $ and any time-translation. They are ergodic  separately with  respect to each of these sub-groups.
\end{lem}

\begin{proof}
We will treat the case of $\mu^A$ only, as the case of $\mu^B$ is similar. By inclusion-exclusion, it is sufficient to consider  an increasing event $E$ depending on rungs in $ \Lambda_{\ell,t}$. Let $ L,k,\ell \in 2\mathbb{N} $ with $L\ge\ell+k$ and $ \beta \geq t + s $.  The comparison between boundary conditions implies that for $ x = (k,s) $,
$$\mu^A_{L+k,\beta+s}[E] \leq \mu^A_{L,\beta} [ \mathcal{S}_{x} E]\le \mu^A_{L-k,\beta-s}[E] .
$$ 
Letting $L,\beta$ tend to infinity implies the invariance under translations.

Any event can be approximated by events depending on rungs in $\Lambda_{\ell,t}$ for some $\ell,t$, hence the ergodicity follows from mixing, i.e.~from the property 
that for any events $E$ and $F$ depending on finite sets,
\begin{equation}\lim_{\substack{|x|\rightarrow \infty \\ x=(k,s) , \,  k \, \mbox{\small even}}}\mu^A[E\cap\mathcal{S}_x F]=\mu^A[E]\mu^A[F].\label{eq:mixing}\end{equation}
Observe that again by inclusion-exclusion, it is sufficient to prove the equivalent result for $E$ and $F$ increasing. 
Let us give ourselves these two increasing events $E$ and $F$ depending on rungs in $\Lambda_{\ell,t}$ only. The FKG inequality and the invariance under translations of $\mu^A$ imply that for sufficiently large $ x = (k,s) $ with $ k $ even:
$$\mu^A[E\cap\mathcal{S}_x F]\ge \mu^A[E]\mu^A[\mathcal{S}_xF]=\mu^A[E]\mu^A[F].$$
In the other direction, fix $\varepsilon>0$ and choose $L=L(\varepsilon)$ and $\beta=\beta(\varepsilon)$ so large than $\mu^{A}_{L,\beta}[E]\le \mu^A[E]+\varepsilon$ and $\mu^{A}_{L,\beta}[F]\le \mu^A[F]+\varepsilon$. If $x=(k,s)$ with $2\ell<k<L-\ell$, then $\Lambda_{\ell,t}$ and its translate by $x$ do not intersect. Thus, the FKG inequality enables to put a unique $A$-cluster in the complement of $\Lambda_{\ell,t}$ and $\mathcal{S}_x\Lambda_{\ell,t}$ disconnecting the two areas so that 
\begin{align*}\mu^A[E\cap\mathcal{S}_x F]\le \mu^{A}_{L,\beta}[E\cap\mathcal{S}_x F]+2\varepsilon\le \mu^{A}_{L,\beta}[E]\mu^{A}_{L,\beta}[F]+2\varepsilon\le \mu^A[E]\mu^A[F]+5\varepsilon.\end{align*}
The result therefore follows by taking $x$ to infinity, and then $\varepsilon$ to 0.
\end{proof}

The second statement is the following important theorem.

\begin{theorem}\label{thm:uniqueness}
For any $ Q \geq 1 $ and $\#\in \{ A, B \} $, one of the two following properties occur:
\begin{itemize}
\item $\mu^\#[(1/2,0)\text{ is $A$-connected to infinity}]=0$ or
\item $\mu^\#[\exists\text{ a unique infinite $A$-cluster}]=1.$
\end{itemize}
\end{theorem} 
Such a result was first proved in \cite{AizKesNew87} for Bernoulli percolation, and was later obtained by other means.  For our needs we shall adapt the beautiful argument of {Burton and Keane} \cite{BurKea89}.  In the proof given below we give only its brief sketch, as its line of reasoning has by now been presented in  many contexts (e.g.~\cite{Dum17}).

\begin{proof} We present the proof for $\mu^A$, since the proof for $\mu^B$ is the same. Let $E_{\le 1}$, $E_n$  ($1<n<\infty$), and $E_{\ge 3}$ be the events that there are no more than one,  exactly $n$,   and finally at least $3$ (possibly infinitely many) infinite $A$-clusters, respectively.     A pair of different arguments will be used to show: i)  $\mu^A[E_n] =0$ for any $1<n<\infty$ (in particular $n=2$),  ii)   $\mu^A[E_{\ge 3}]=0$.  This  leaves: $\mu^A[E_{\le 1}]=1$.

 Assume  that  $\mu^A[E_{n}] >0$  for some $1<n<\infty$.  Then there exist 
$\ell$ and $t$ large enough so that $\mu^A[F_{\ell,t}]\ge\tfrac12\mu^A[E_{n}]>0$, where $F_{\ell,t}$ is the event that all the infinite $A$-clusters in $(\mathbb Z\times\mathbb R)\setminus \Lambda_{\ell,t}$ (if there are any) intersect $\Lambda_{\ell,t}$. 
The event $F_{\ell,t}$ is independent of the rungs in $\Lambda_{L,\beta}$ and conditioned on it there is a positive probability of the event $G_{\ell,t}$ that all the boundary vertices of $\Lambda_{L,\beta}$ are $A$-connected in $\Lambda_{L,\beta}$.  Hence
\be \label{n}
\mu^A[E_{\le 1}]\ge\mu^A[F_{\ell,t}\cap G_{\ell,t}]>0\,. 
\ee 
However, by the translation invariance of the event $E_{\le 1}$ and the ergodicity of the infinite volume probability distribution, $\mu^A[E_{\le 1}]$ can take only the values $0$ or $1$.   Therefore \eqref{n} implies 
that $\mu^A[E_{\le 1}]=1$, and thus $\mu^A[E_n=0]$, contradicting the assumption.

To prove that $\mu^A[E_{\ge 3}]=0$ we consider trifurcation events, along the lines of Burton-Keane.   A  trifurcation event $\mathcal T^{(\ell)}_{n,m}$, of scale $\ell$, is said to occur within the box $B^\ell_{(n,m)} = [n\ell,(n+1)\ell]\times [m\ell,(m+1)\ell] \subset \R^2$ if for the given $\omega$ there exists a  point within $B^\ell_{(n,m)}$  which is connected to infinity by three paths among which there is no connection outside the box. 

Assume $\mu^A[E_{\ge 3}]>0$.  Then for $\ell$ large enough with positive probability (which tends to $\mu^A[E_{\ge 3}]$ for $\ell \to \infty$)   the box $B_\ell=[0,\ell]\times [0,\ell]$ intersects three distinct infinite clusters.  

It is easy to see that for any exterior configuration (with locally finite edge set)   for which this condition is met,  
 there exists a non-empty \emph{open} set of interior configurations  for which there is a trifurcation within $B_\ell$.   Since the conditional probability of any non-empty open  set of configurations   is strictly positive (our analog of BK's ``finite energy'' condition)  one may conclude that  for large enough $\ell$ 
\be \label{positivity}
\mu^A[\mathcal T^{(\ell)}_{0,0}]\ \geq \  C_\ell \, \,  \mu^A[E_{\ge 3}]
\ee
with $C_\ell >0 $  (by further inspection that extends to all  $\ell \geq 1$, but  this refinement is not necessary). 

By translation invariance the mean number of trifurcation events of scale $\ell$ which occur within the 
finite region $ \Lambda_{L,T} = [0,L]\times [0,T]$  in translates of $B_\ell$ by $(n,m)\in (2\ell \mathbb Z)^2 $   increases in proportion to the volume:  
\be \label{volume} 
\mathbb{E}^A(N) \  \geq \ \mu^A[\mathcal T^{(\ell)}_{0,0}] \, {L \cdot T }/  (4 \ell^2)\,. 
\ee

The  Burton-Kean argument is to combine this with the observation that in any configuration with  
$N$ such  trifurcation events within $\Lambda_{L,T}$ there need to be  at least $N$  distinct infinite $A$-clusters intersecting the boundary of that set.   However, this number cannot grow faster than the boundary.  More explicitly  
\be  \label{boundary}
\mathbb{E}^A(N)  \ \leq \ 2\, T \, \mathbb{E}^A[\mathsf N_1]   \ +\  2\, L \,  \mathbb{E}^A[\mathsf N_2]   \,,     
\ee
where  
$\mathsf N_1$ is the number of distinct $A$-clusters of the half space $(-\infty, 0] \times \mathbb R$ reaching  $[0,1]\times [0,1]$,  $\mathsf N_2$ is the number of distinct $A$-clusters of the half space 
$\mathbb R\times (-\infty,0)$  reaching  that set, and  $\mathbb{E}^A(-)$ denotes the $\mu^A$ expectation value.

By  elementary (local) estimates $\mu^A[\mathsf N_j]<\infty$, for both $j=1,2$ (as  distinct $A$-clusters require the crowding of separation events).   Combining \eqref{volume} with \eqref{boundary}, and letting $L$ and $T$ tend to infinity at comparable speeds one  learns that for all  $\ell < \infty$
\be \label{tau}
\mu^A[\mathcal T^{(\ell)}_{0,0}]\ =\ 0\,.   
\ee
That is in contradiction with \eqref{positivity} (derived under the assumption that  $\mu^A[E_{\ge 3}]>0$), which proves our claim. 
\end{proof}

We shall also use the following statement.
 \begin{lem}\label{lem:Zhang}
All $B$-clusters are finite $\mu^A$-almost surely.
\end{lem}
By duality, $\mu^B$-almost surely all the $A$-clusters are finite almost surely. 
\begin{proof}
Assume by contradiction that $\mu^A[\exists\text{ infinite $B$-cluster}]=1$ and fix $L,\beta$ so that \begin{equation}\label{eq:h1}
\mu^A[\Lambda_{L,\beta}\text{ is $B$-connected to infinity}]\ge 1-\tfrac{1}{10^4}
\end{equation} 
and the probability that the top, bottom, left and right of the boundary of $\Lambda_{L,\beta}$ are $B$-connected to infinity in the complement of $\Lambda_{L,\beta}$ are the same (simply fix $L_0,\beta_0$ large enough to get \eqref{eq:h1}, and then increase $L_0$ and/or $\beta_0$ in order to obtain $L,\beta$).

Since a path from infinity to $\Lambda_{L,\beta}$  ends up either on the top, bottom, left or right of it, the FKG inequality implies (through the square-root trick\footnote{The square-root trick refers to the $k=2$ case fo the observation  that for every increasing events $A_1,\dots,A_k$, the FKG inequality implies that 
 $\max\{\mu[A_i] | 1\le i\le k\}\ge 1-\mu[\left(A_1\cup\dots\cup A_k\right)^c]^{1/k}$. 
This inequality is an improvement on the union bound, as it shows that if the union of $k$ increasing events has a probability close to 1, then this is also true for at least one of these events.}) 
that 
$$\mu^A[\text{top of $\Lambda_{L,\beta}$ $B$-connected to infinity outside $\Lambda_{L,\beta}$}]\ge 1-\tfrac1{10^{4/4}}=1-\tfrac1{10},$$
and similarly for the right, bottom and top.  
Now, assume that the top and bottom are $B$-connected to infinity, and the left and right (more precisely the $A$-lines on the left and right of the respective boundaries) are $A$-connected to infinity (the probability of the latter is larger than the probability that there exists a $B$-connection since under $\mu^A$ the $A$-clusters dominate the $B$-ones\footnote{By symmetry the $B$-clusters under~$\mu^A$ are distributed as the $A$-clusters under~$\mu^B$, which by FKG is dominated by~$\mu^A$.}). The union bound implies that this happens with probability $1-\tfrac4{10}>0$.
Yet, the finite-energy property also implies that conditionally on this event, the rungs in $\Lambda_{L,\beta}$ are such that no boundary vertex of $\Lambda_{L,\beta}$ are $A$-connected using paths in $\Lambda_{L,\beta}$, implying that there exist two infinite $A$-clusters with positive $\mu^A$-probability. But this contradicts the fact, proved in the previous statement, that there is zero or one infinite $A$-cluster.\end{proof}

We are now in a position to prove  this section's main result.
\begin{proof}[Proof of Theorem~\ref{thm:mainpercolation}]
Property 1 is a direct consequence of the fact that $\mu^A$ does not possess any infinite $B$-cluster. It also means that when fixing a finite set, and taking $L$ and $\beta$ large enough, no loop intersecting the finite set reaches the boundary of $\Lambda_{L,\beta}$ or winds around the vertical direction. By construction, we deduce that all these loops are oriented in an independent fashion described in the previous section. As a consequence, Properties 2 and 3 follow trivially. Property 4 is a direct consequence of $\mu^A=\mu^B$, so that the distribution of the $A$- and $B$-clusters is the same under $\mu^A$. In particular, there is no infinite $A$- or $B$-cluster, which immediately implies Property 4. Finally,  if the two measures are equal,  $\mu^{\mathrm{per}}_{L,\beta}$ stochastically dominates $\mu^B_{L,\beta}$ and is stochastically dominated by $\mu^A_{L,\beta}$. Since these measures converge to the same measure $\mu^A=\mu^B$, so does $\mu^\mathrm{per}_{L,\beta}$.
\end{proof}

\section{Proofs of symmetry breaking}\label{sec:percolation}

We now have the tools  for a structural proof of the different forms of symmetry breaking in the  models considered here.  We start with the translation symmetry breaking in the limiting distribution of the  random loop measure  for all $Q>4$.   This is then used to conclude dimerization in the ground-states of the $H_\textrm{AF}$ spin chains with $S>1/2$, and N\'eel order in the ground-states of spin $1/2$ XXZ-chain  at $\Delta >1$.  These results  can be obtained through the rigorous analysis of the Bethe ansatz along the lines of Duminil et.al.~\cite{DGHMT}, which also yields more quantitative information.  However, for a shorter and somewhat more transparent proof we present 
an analog of the argument which was recently developed  by Ray and Spinka~\cite{RS} in the context of the $6$-vertex/$Q$-state random-cluster model on $\Z^2$.

\subsection{Translation symmetry breaking for the loop measure at $Q>4$}

As in \cite{GlazPel19,RS} we shall make an essential use of a random height function $h: (\R\setminus \Z) \times \R \mapsto \Z$,  which in our case is assigned the configurations of $\vec\omega = (\omega,\tau) $.   The function is piecewise constant with discontinuities at lines supporting the loops of $\omega$.  
Along the horizontal line $t=0$ it is defined by:  
\be
h_{ \vec\omega }(-1/2,0) := 0 , \qquad 
h_{ \vec\omega }(u,0) :=\begin{cases} 
 \displaystyle \sum_{n\in (-1/2, u)\cap \Z} \tau(n,0) & \mbox{if $u \geq 0$,} \\  \displaystyle \sum_{n\in (u,-1/2)\cap \Z} -\tau(n,0)  &  \mbox{if $u < -1 $} \, .
\end{cases} 
\ee 
More generally, the value of $h_{ \vec\omega }(u,t)$ at any point  off the  loop lines of $\omega$ 
is the sum of the fluxes of~$\tau$ across an arbitrary simple path from $(1/2,0)$  to $(u,t)$,  counted with the sign of the cross product of the direction of $ \tau $ with the curve's tangent at the point of crossing.  In this description of the height function one may restrict the attention to paths which avoid rungs, i.e. which crosses loops only at vertical boundaries.   (To avoid misunderstanding, let us add  that  when the  crossing occurs at horizontal rungs of $\omega$ the $h$-function  increases by $0$ or $\pm 2$, depending on the orientations of the two loops which cross the rung and the direction of crossing.)   An example of this construction is presented  in Fig.~\ref{fig:Height}.

\begin{figure}[h]
\begin{center}
\includegraphics[height=2.5 in]{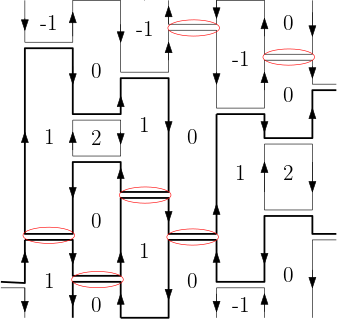}
\caption{An example of the height function in relation to a specified configuration of rungs and loop orientations (indicated by the arrows).  One may note that from the discontinuities of  the pseudo spin $\tau$ one can deduce the presence of some of the loops of $ \omega$.  However, $\omega$ also includes other rungs (marked here in small ovals) whose presence is transparent to $\tau$.  Their knowledge is essential for the full reconstruction of the loops of $\omega$, but not  for the height function.  An implied combination of properties of the height function (Lemma~\ref{lem:contradiction}) plays an essential role in the analysis of the case $Q>4$.}\label{fig:Height}\end{center}
\end{figure}

As a random function,  $h(x,t)$ exhibits a surprising combination of properties: 
\begin{enumerate} 
\item  the function's statistical distribution is simplest to present in the language of the loops, whose structure  depends on all the loops of $\omega$ (i.e., the erasure of  any rung does change the loop structure).
\item  
yet, for any $(\omega,\tau)$ the values of  $h_{ \vec\omega }(u,t)$ at generic points  can be read from just the pseudo spin function $\tau$ (i.e. it does not require the full knowledge of $\omega$). 
\end{enumerate} 

More explicitly, we have the following auxiliary statement.  
What makes it particularly astounding is the combination of the above  with the third assertion listed below.  


\begin{lem}   \label{lem:contradiction} For  any $Q>4$, let $ \widehat\mu$  be a probability measure on the systems of oriented loops   over $\Z\times \R$, described by the variables $  \vec\omega = (\omega, \tau) $, with the properties: 
\begin{enumerate} 
\item $ \widehat\mu$ almost surely the loop configurations corresponding to $\omega$ consist of only finite loops.
\item conditioned on $\omega$ the loops are oriented independently of each other clockwise ($-$) or counterclockwise $(+)$, which the probabilities $e^{\pm\lambda}/[e^{-\lambda} +e^{+\lambda}]$, correspondingly, 
\item the marginal distribution of $\tau$ variables is invariant under a global reversal of orientation ($ \lambda\mapsto -\lambda$). 
\end{enumerate} 
Then the event 
\begin{equation}
\mathcal{N} := \left\{ \mbox{ $ (1/2,0) $  is encircled by an infinite number of  loops of $\omega$} \right\} 
\end{equation}
has zero measure.
\end{lem}
\begin{proof}  The proof is by contradiction.  Let $\alpha_k = (u_k,t_k)$ be a sequence of points with 
$|\alpha_k|  \nearrow    \infty$  each of which lies on a level set of $h_{ \vec\omega }$ which includes a path winding around $(-1/2,0)$.  
The values of $h(\alpha_k)$ are given by the sum of loop orientations over those loops of $\omega$ which separate  
$\alpha_k$ from $(-1/2,0)$.  For $Q>4$, these orientations are given by a sequence of iid $\pm1$ valued  random variables, of the non-zero mean $\tanh(\lambda)$.   It readily follows that almost surely the following limit exists and satisfies  
\be
\lim_{k\to \infty} h_{ \vec\omega }(\alpha_k) = \begin{cases} \sgn(\lambda) \, \infty & (\omega,\tau) \in \mathcal{N} \\ 
\mbox{a finite value} & (\omega,\tau) \notin \mathcal{N}
\end{cases}  \ .
\ee 
From this one may deduce that the probability distribution of  $\lim_{k\to \infty} h(\alpha_k) $ is: 
\be \label{eq:contra}
\lim_{k\to \infty} h_{ \vec\omega }(\alpha_k)  = \begin{cases} \sgn(\lambda) \, \infty & \mbox{with probability $\widehat\mu(\mathcal{N})$} \\  
\mbox{a finite value}  &   \mbox{with probability $[1- \widehat\mu(\mathcal{N})]$ }
\end{cases}\ .
\ee 
The level sets of  $h_{ \vec\omega }$ can be determined from just $\tau$ (i.e. ignoring $\omega$)  and, under the assumption  made here its distribution does not change under the flip $\lambda \to -\lambda$.  Hence \eqref{eq:contra} yields a contradiction unless   $\widehat\mu(\mathcal{N})=0$. 
 \end{proof} 

From this we shall now deduce three symmetry breaking statements.  The first concerned just the loop measure, but in the statement's proof we make use of the measure's significance for the two quantum spin models. 
\begin{theorem} \label{thm:A_B_breaking}
For all $Q >4$ (equivalently $S>1/2$) the loop measures corresponding to the even and odd ground-states  $\langle \cdot\rangle_{\text even} $ and $\langle \cdot\rangle_{\text odd} $ of \eqref{ev_odd}  differ.   
\end{theorem} 
\begin{proof}  
Assume the two loop measures coincide.   Then by Theorem~5.1 (4) the event $\mathcal N$ occurs with probability $1$ with respect to the common probability measure. 

However, by Theorem~5.1 (5), under the above assumption these measures also describe the limiting distribution of 
$\omega$ corresponding to the ground-state of the periodic operator  $H^{(L,per)}_{XXZ}$ in  the limit $L\to \infty $, which for concreteness sake we take to be along $2\Z$ and with  $\lambda $ being the positive  solution of \eqref{lambda_Q}.
However,  as was noted in Corollary~\ref{cor:symm}  the periodic boundary condition state is actually an even function of $\lambda$.    It follows that the limiting state satisfied all the assumptions made in Lemma~\ref{lem:contradiction}, and hence the event  $\mathcal N$ is of probability zero.

The  contradiction between the two implications of the above assumption  implies that  the measures are distinct. 
\end{proof}

\subsection{Dimerization for $ S>1/2$ }

Next, we extract from the above probabilistic statement a proof of dimerization in the quantum $H_{AF}$ spin chain.

\begin{proof} [Proof of Theorem~\ref{thm:main1}]   
1.~From Theorem~\ref{thm:A_B_breaking} we already know that for any $S>1/2$ the loop measures associated with the states $ \mu^A \equiv \langle \cdot\rangle_{\text even} $ and $\mu^B \equiv  \langle \cdot\rangle_{\text odd} $  are different. 
Theorem  \ref{thm:uniqueness} allows to identify the difference in percolation terms: in both cases there is almost surely a unique infinite connected cluster, which is of type $A$ in one and $B$ in the other case.  More explicitly,  
\be 
\mu^A\left[  (u+1/2,0)  \leftrightarrow \infty  \right] = 
 \mu^B\left[  (u-1/2,0) \leftrightarrow \infty  \right]   = \begin{cases} 
0   & \mbox{ for $u$ odd} \\ 
p_\infty>0  & \mbox{ for $u$ even} 
\end{cases}  \ .
\ee  

However, that leaves still the challenge to determine whether this difference between the two measures and, in each, between the even and odd sites, can be detected in terms of  a  physical observable, i.e.~the expectation value of some function of the spin degrees of freedom.  
 
The  question was addressed in  \cite{AN} where it is shown (see also the next section for a similar reasoning) that: \\  
\indent i)  since the two limiting distributions of $\omega$ are related by the FKG inequality, if they differ then the difference is also manifested in the more elementary connectivity probability to lie on the same loop, and in particular for even sites $u=2n$:
\be \label{conn_ineq} 
\mu^A\left[  (2n,0) {\leftrightarrow}  (2n+1,0) ] \right]  \ < \   
\mu^A\left[ (2n-1,0) {\leftrightarrow}  (2n,0) ]  \right].
\ee 
\indent ii) combining \eqref{conn_ineq} with   \eqref{eq:proexp} one learns that
\be 
\langle  P_{2n,2n+1}^{(0)} \rangle_{\text even}  \ <  \ \langle   P_{2n-1,2n}^{(0)} \rangle_{\text even} 
\ee 
with the opposite inequality for the odd state.    Thus the  differences in the two states are detectable and extensive. \\

\noindent 
2.~Theorem~7.2 of \cite{AN} bounds the truncated quantum correlations from above in terms of the correlations 
\begin{equation}\label{eq:ha}
\min_{\# \in \{ A,B \} }\mu^\#[\mbox{$(\pmb{u},t) $ and $ (\pmb{v},0) $ belong to the same cluster}] 
\end{equation}
of the underlying loop measures for every $ \pmb{u} \in U \pm 1/2 $ and $ \pmb{v} \in V \pm 1/2 $. Now, the percolation model considered in Section 1.3 of \cite{DLM18} corresponds exactly to our model here.\footnote{One may be surprised by the fact that the Poisson point process there has intensity 1 and $q$ depending on the column. This comes from the fact that the Radon-Nikodym derivative is expressed in \cite{DLM18} as $q$ to the number of $A$-clusters, which can be shown, using Euler's formula, to be expressed in terms of $\sqrt q$ to the number of loops if one change the intensity of the Poisson point process to 1 in every column.} The fifth bullet of Theorem.~1.5 of \cite{DLM18} thus implies, with the notation of our paper, that there exists $c=c(q)>0$ such that for every $n\ge1$,
\begin{equation}
\mu^B[\text{$(\tfrac12,0) $ belongs to a $A$-cluster reaching $\partial\Lambda_{n,n}$}] \le \exp[-cn].
\end{equation}
Since for $ \pmb{u} \in U \pm 1/2 $ and $ \pmb{v} \in V \pm 1/2 $ to be connected to each other, there must exist a path from $\pmb{u}$ to the translate of $\partial\Lambda_{n,n}$ by $\pmb{u}$, where $n=\max\{|t|,|\pmb{u}-\pmb{v}|-1\}$, we deduce that all the quantities in \eqref{eq:ha} are smaller than $C\exp[-(\dist(U,V)+|t|)/\xi]$ for some small enough constant $\xi=\xi(q)>0$.
\end{proof} 


\subsection{N\'eel order  for  $\Delta=\cosh(\lambda)>1$}\label{App:Neel}
 
Turning to the ground-states of the XXZ-spin system, let us recall the notation.  Let $\lambda>0$ be the positive solution of \eqref{Delta_lambda}, and  denote by 
$\widehat\mu^{A}_{+ \lambda}  $ the even limit (cf. \eqref{eq:evenoddlimit}) of the measure  
on the enhanced system of the variables $ (\omega,\tau)$, in which the winding probabilities of loops of $\omega$ are $e^{\pm\lambda}/(e^\lambda + e^{-\lambda} ) $, with   $(+)$ winding corresponding to the counterclockwise and $(-)$ the clockwise orientation.   The corresponding state  on (just) the spin variables, is denoted $\langle  \tau(u,0)  \rangle_+$.   These are to be contrasted with $\widehat\mu^{A}_- $  and $\langle  \tau(u,0)  \rangle_-$.  The superscript may be omitted, but it should be remembered that  in both $+$ and $-$ cases, the limit $L\to \infty$ is taken over the even sequence.

 \begin{proof} [Proof of Theorem~\ref{thm:main2}]    
 As we saw,  the probability distribution of $\omega$ under the two measures $\widehat\mu^{A}_{\pm \lambda}  $ agree with that of the corresponding $H_{AF}$ system.  Hence, in both cases there is positive probability $  p_\infty>0  $ that  $  (1/2,0) $ belongs to an infinite connected cluster.  When that happens,  the sign of $\tau(0,0)$ coincides with the winding sign of the  loop which passes through $ (0,0)$.   Thus the spin $ \tau(0,0) $  takes the values $ \pm $ with probabilities $ e^{\pm\lambda}/(e^\lambda + e^{-\lambda} ) $.  The same holds true for any even site $u=2n$. Consequently,   for any $n\in \Z$:
\begin{align}   \label{SB_AF}
\int  \tau(2n,0) \cdot \1[  (2n+1/2,0)\leftrightarrow \infty ]\,  \widehat\mu^{A}_{+} (d\vec{\omega})   = & \  \tanh (\lambda)  \, p_\infty    \notag  \\ 
 =& \ -  \int  \tau(2n,0) \cdot \1[  (2n+1/2,0)\leftrightarrow \infty ]\,  \widehat\mu^{A}_{-} (d\vec{\omega})  \, .
\end{align} 
Since  $p_\infty > 0$, we deduce that $\widehat\mu^{A}_+ $ and $\widehat\mu^{A}_- $ are different.  

To that let us add the observation that percolation with respect to $\omega$ corresponds to percolation along a level set of the height function which is readable from $\tau$.  Hence the observable which distinguishes the two states is in principle a functional of the physically meaningful spin function. 
We conclude that  also the states  $\langle  \cdot  \rangle_+$  and $\langle  \cdot  \rangle_-$ are different, and hence the  infinite XXZ-spin system has at least two different ground-states.

A remaining challenge is to simplify the distinction between the two states, as was done in step (ii) of the above proof of Theorem~\ref{thm:main1}.   For the XXZ-model that can be deduced using the last statement in Proposition~\ref{prop:Neel}.  It 
allows to conclude  from \eqref{SB_AF}  (and the previously established fact that 
$p_\infty >0$) that for any $u\in \Z$
\be\label{mag_diff}
\langle  \tau(u,0)  \rangle_+ \neq \langle  \tau(u,0)  \rangle_{-}  \,.  
\ee  
This also implies the non-vanishing of  the N\'eel order parameter $M$ of \eqref{alt}. 
\end{proof}

In the last step, leading to \eqref{mag_diff},  we invoked an   ``FKG boost'' whose full discussion was postponed in order to streamline the presentation of the main results.  Following is its proof.     
\begin{proof}[Proof of Proposition~\ref{prop:Neel} ] 
The convergence of each of the two finite-volume ground-states in the limit $L \to \infty$, with $L$ limited to even values, is based on the FKG monotonicity of the percolation model discussed in the previous section, which is common to the two systems discussed here.  

It remains  to establish that for the  XXZ system at any $\Delta \geq 1$, the  states $\langle  \cdot  \rangle_+ $ and $\langle  \cdot  \rangle_-$ are either equal or of different magnetization, satisfying \eqref{mag_diff}.  
In the  proof we shall again employ the  FKG inequality but do so in a different setup than used above.  This time it will be in the context of an Ising-like representation of the distribution of the staggered spins
\be \label{kappa} 
\kappa(u,t) \ = \ (-1)^u \tau(u,t)  \, .  
\ee  
In terms of these variables the ground-states of XXZ-system take the form of an annealed Gibbs equilibrium state of a ferromagnetic Ising model, and the two boundary conditions correspond to $(+)$ and correspondingly $(-)$ fields applied along the ``vertical'' part of the boundary.    

To present  the ground-states in this form we start with  the following preparatory steps:  
\begin{enumerate} 
\item[1)] Rewrite the XXZ Hamiltonian of \eqref{mod_XX}, 
with the boundary term of \eqref{XXZ_bc}, 
as the spin-$1/2$ version of the $H_{AF}$ operator with an added anti-ferromagnetic coupling of strength $\delta= \Delta -1$.  I.e., for any even $L$:
  \begin{equation} \label{mod_XX2}
U_L   H^{(L, \pm)}_{\textrm{XXZ}}U_L ^* = 
\frac{1}{2} \sum_{v=-L+1}^{L-1}\,  
\left[    \pmb{\tau}_{v}  \cdot \pmb{\tau}_{v+1}   +  \pm  \delta \tau^z_{v}  \tau^z_{v+1}  - \Delta \cdot \1\right]   
+ \frac{ \sinh(\lambda) }{2} [   \tau_{-L+1}^z - \tau_L^z ]  
\,. 
\end{equation} 
\item[2)]  For each even $L$ and choice of the $\pm$ sign,  construct the  ground-state of this operator,  applying  the hybrid representation \eqref{Poisson3}, with the $\delta $ term treated as a potential ($V$) added to  the  $\Delta=1$ operator, and starting   
from the seed state $ | N_0^{(L)} \rangle $ of \eqref{N_vector}  with $ \lambda = 0 $.  
\end{enumerate} 

We are  particularly interested in operators which arise from  functionals of $ \tau $, or equivalently of $\kappa $ (the two being related by \eqref{kappa}).   Taking the limits indicated in \eqref{eq:101}, we get 
 \begin{eqnarray}
\langle F \rangle_{\pm} &=&  
\lim_{\substack{L\to \infty \\ L \textrm{ even}}}   \lim_{\beta\to \infty} \langle F \rangle_{L,\beta,0}^{(\textrm{XXZ},\pm)} 
  \end{eqnarray}
with the  Feynman-Kac type functional integral 
\begin{eqnarray}
\langle F \rangle_{L,\beta,0}^{(\textrm{XXZ},\pm)}   = \ \frac{ \int  \mathbb{E}_{L,\beta}^\pm\left[ F| \omega \right]   \, Z_{L,\beta}^\pm (\omega) \, \rho_{L,\beta}(d\omega)  }{\int  Z_{L,\beta}^\pm (\omega) \, \rho_{L,\beta}(d\omega)} 
 =:\int  \mathbb{E}_{L,\beta}^\pm\left[ F| \omega \right]   \, \mu_{L,\beta}(d\omega)  
\end{eqnarray}
with 
\begin{align}
Z_{L,\beta}^\pm (\omega) \ & := \sum_\kappa  \1[\omega,\kappa]  \  \exp\Big(  - \int_{-\beta/2}^{\beta/2} V_\pm(t) \, dt  \Big) \, , \notag \\  
 \mathbb{E}_{L,\beta}^\pm\left[ F | \omega \right]  \  & := \frac{1}{Z_{L,\beta}^\pm (\omega)} \sum_\tau  \1[\omega,\kappa]  \  \exp\Big(  - \int_{-\beta/2}^{\beta/2} V_\pm(t) \, dt  \Big) \ F[\kappa] \, , 
 \end{align}
 where, in a slight abuse of notation, we denote by $\1[\omega,\kappa]$ the consistency indicator function which is inherited from $\1[\omega,\tau]$ through the correspondence \eqref{kappa},  and the potential is
 \be
V_\pm(t)  \ := - \sum_{u=-L+1}^{L-1} \kappa(u,t) \ \kappa(u+1,t) \mp \frac{\sinh(\lambda)}{2}  \left[\kappa(-L+1,t)  +\kappa(L,t) \right] \, . 
\ee

It is now important to note that in terms of $\kappa$ the consistency condition translated into the constraint that $\kappa(u,t)$  is constant along each of the loops of $\omega$.   Thus, for a given $\omega$ the function $\kappa(u,t)$ is fully characterized by the collection of binary variables $\{ \kappa_\gamma \}$ indexed by the loops of $\omega$.  Furthermore, the potential   $V_\pm(t)$ is expressible as a ferromagnetic pair interaction among the values of this collection of observables.   

By the general FKG property of the ferromagnetic Ising measures, for each $\omega$ the $(+)$ and $(-)$ measures 
$ \mathbb{E}_{L,\beta}^\pm\left[ \cdot | \omega \right]$  admit a monotone coupling (this is referred to as Strassen's theorem). That is, there exists a joint probability distribution $ \hat \nu_\omega(\kappa_+,\kappa_-) $, with marginals given by the two measures $ \mathbb{E}_{L,\beta}^\pm\left[ \cdot | \omega \right]$  and which is supported on pairs of configurations satisfying 
\be
	\kappa_+(u,t) \geq \kappa_-(u,t)  \quad  \forall   (u,t)\in \Lambda_{L,\beta} \, . 
\ee
In terms of such a coupling, for any   functional $F[\kappa]$: 
\be
  \mathbb{E}_{L,\beta}^+\left[F[\kappa]\, | \omega \right]  -  \mathbb{E}_{L,\beta}^-\left[ F[\kappa] \,  | \omega \right] = 
  \sum_{\kappa_+,\kappa_-}  \left( F[\kappa_+] -  F[\kappa_-]\right) \hat\nu_\omega(\kappa_+,\kappa_-) \, . 
\ee
and for $F[\kappa] $ monotone non-decreasing the terms on the right are all non-negative.  

In studying the limit $L,\beta \to \infty$ it is convenient to measure the distance of the two induced measures within rectangular space-time domains  of the form $B_{K,T}=[-K,K]\times [-T,T]$, through the Wasserstein-type  metric
\be
W_{K,T}\big(\langle \cdot \rangle_{L,\beta,0}^{(\textrm{XXZ},+)}, \langle \cdot \rangle_{L,\beta,0}^{(\textrm{XXZ},-)}\big)
 =  \int \Big\{ \inf_{\nu_\omega} \Big[ \sum_{u=-K}^K \int_{-T}^T |\kappa_+(u,t) -\kappa_-(u,t)| \,dt \Big]  \nu_\omega(\kappa_+,\kappa_-) \Big\} \mu_{L,\beta}(d\omega) 
\ee 
where $\nu_\omega$ ranges over couplings of the two measures $ \mathbb{E}_{L,\beta}^\pm\left[ \cdot | \omega \right]$.  

For the monotone coupling the absolute value can be dropped, in which case the integral reduces to the simple difference in expectation values of $\kappa$, and thus 
\be
W_{K,T}\big(\langle \cdot \rangle_{L,\beta,0}^{(\textrm{XXZ},+)}, \langle \cdot \rangle_{L,\beta,0}^{(\textrm{XXZ},-)}\big)
 =   \sum_{u=-K}^K \int_{-T}^T  [\langle \kappa(u,t) \rangle_{L,\beta,0}^{(\textrm{XXZ},+)}  -\langle \kappa(u,t) \rangle_{L,\beta,0}^{(\textrm{XXZ},-)}  ] \,dt   \,. 
 \ee  

In the infinite volume limit, in which the mean value of $\kappa(u,t)$ is translation invariant, one gets
\begin{eqnarray}  
\lim_{\substack{L\to \infty \\ L \textrm{even}}}   \lim_{\beta\to \infty} W_{K,T}\big(\langle \cdot \rangle_{L,\beta,0}^{(\textrm{XXZ},+)}, \langle \cdot \rangle_{L,\beta,0}^{(\textrm{XXZ},-)}\big)
 &=&   |B_{K,T}|  \left[  \langle \kappa(0,0) \rangle_+  -\langle \kappa(0,0) \rangle_-  \right]  \notag \\    &=&  4 (2K+1) \,T \, \langle \tau(0,0) \rangle_+ \,.
\end{eqnarray} 

We learn that if $\langle \tau(0,0) \rangle_+  =0$  then for any  $K,T<\infty$  the Wasserstein distance between the restrictions of the two measures to the box $B_{K,T}$ tends to zero.   
It follows that if the measures converge, as we know to be the case here,   then they have a common limit. 

In other words, if  for some measurable functional of $\tau$, $\langle  F \rangle_+ \neq \langle F\rangle_-$, then we may conclude that also 
$
\langle \tau(0,0) \rangle_+  \neq \langle \tau(0,0) \rangle_-  \,. 
$
In view of the relation between the two states the latter is equivalent to 
\be 
\langle \tau(0,0) \rangle_+  \neq 0 \,.
\ee
  \end{proof}

\section*{Postscript -- quantum degrees of freedom as emergent features}

The analysis presented here provides another example where the categorical distinction between classical and quantum physics is  blurred. 
We started with two quantum spin chains and moved on to their relation with a common random loop model.   An alternative presentation could have started from the random loop model, based on the random rung configurations $\omega$,
 which is of independent interest in probability theory and statistical mechanics and then proceed by recognising that this system's features can be best understood through emergent quantum degrees of freedom.  

The utility of such crossings of the quantum/classical divide has been noted before: 
 In one direction,  the thermodynamic of the planar Ising model are best explained in terms of emergent quantum degrees of freedom, among which are Bruria Kaufmann's spinors~\cite{kau} and Lieb-Mattis-Schultz fermions~\cite{LMS}.  In the other direction one finds  Feynman-Kac functional integral representations for thermal states of quantum particle system in terms of classical functional integrals, and analogous formulas for  quantum spin chains, such as employed  in~\cite{Fey, Gin, AL, T, AN,Uel,BU}.  \\

\minisec{Acknowledgments}
We thank Ron Peled and Yinon Spinka, Edward Witten and Bruno Nachtergaele for  stimulating discussions and relevant references. 
MA is supported in parts by the NSF grant DMS-1613296, and the Weston Visiting Professorship at the Weizmann Institute of Science. HDC was supported by  the ERC CriBLaM, the NCCR SwissMAP, the Swiss NSF and an IDEX Chair from Paris-Saclay.
SW is supported by the DFG under EXC-2111 -- 390814868.

\end{document}